%% file: main.tex
\documentclass[final]{IEEEtran}
\usepackage{color}
\usepackage{csquotes}
\usepackage{enumitem}
\usepackage{float}
\providecommand{\tabularnewline}{\\}
\floatstyle{ruled}
\newfloat{algorithm}{tbp}{loa}
\providecommand{\algorithmname}{Algorithm}
\floatname{algorithm}{\protect\algorithmname}
\usepackage{url}

\makeatletter
\newcommand{\manuallabel}[2]{\def\@currentlabel{#2}\label{#1}}
\makeatother
\usepackage{amsmath}
\usepackage{amssymb}
\usepackage{amsthm}
\usepackage{bbm}
\usepackage{tikz}
\usepackage{mathrsfs}
\usepackage{pgfplots}
\usepackage{cite}
\usepackage{hyperref}

\pgfplotsset{compat=1.14}
\usetikzlibrary{arrows}
\usetikzlibrary{arrows.meta}
\newtheorem{theorem}{Theorem}
\newtheorem{lemma}[theorem]{Lemma}

\newtheorem{dfn}[theorem]{Definition}

\newtheorem{cor}[theorem]{Corollary}

\newcommand{\rbar}{R_0}
\newcommand{\rin}{R_{\mathrm{in}}}
\newcommand{\supp}{\operatorname{supp}}

\title{Exact Error Exponents of \\ Concatenated Codes for DNA Storage} 

\author{Yan Hao Ling and Jonathan Scarlett\thanks{The authors are with the  Department of Computer Science, School of Computing, National University of Singapore (NUS). J.~Scarlett is also with the Department of Mathematics, NUS, and the Institute of Data  Science, NUS. Emails: \url{lingyh@nus.edu.sg};  \url{scarlett@comp.nus.edu.sg}}\thanks{This work was supported by the Singapore National Research Foundation (NRF) under grant number A-0008064-00-00.}}

\begin{document}
\maketitle

\begin{abstract}
    In this paper, we consider a concatenated coding based class of DNA storage codes in which the selected molecules are constrained to be taken from an ``inner'' codebook associated with the sequencing channel.  This codebook is used in a ``black-box'' manner, and is only assumed to operate at an achievable rate in the sense of attaining asymptotically vanishing maximal (inner) error probability.  We first derive the exact error exponent in a widely-studied regime of constant rate and a linear number of sequencing reads, and show strict improvements over an existing achievable error exponent.  Moreover, our achievability analysis is based on a coded-index strategy, implying that such strategies attain the highest error exponents within the broader class of codes that we consider.  We then extend our results to other scaling regimes, including a super-linear number of reads, as well as several low-rate regimes.  We find that the latter comes with notable intricacies, such as dependencies of the error exponents on the model for sequencing errors.
\end{abstract}

\input{isit_part.tex}

\section{Super-Linear Number of Reads} \label{sec:superlinear}

In practical scenarios, the number of reads performed may be large, since performing reads is relatively cheap compared to other steps (e.g., synthesis).  One way to study the regime of a large number of reads is to let $c$ grow large in Theorem \ref{thm:main_result}, as we did previously in  \eqref{eq:lim_c_large}, but it is also of interest to understand how the error probability behaves when $N = \omega(M)$, i.e., a super-linear number of reads.  The following result shows that the error exponent is particularly simple in this case.

\begin{theorem} \label{thm:superlinear}
    Consider the setup of Theorem \ref{thm:main_result}, except that we have $N = \omega(M)$ instead of $N = cM$.  Then, the error exponent \emph{with respect to $N$} is simply $\log \frac{1}{R_0}$ in the sense that
    \begin{equation}
        \lim_{M \to \infty} \frac{-\log P_e^*(M)}{N} = \log \frac{1}{R_0}.
    \end{equation}
\end{theorem}

Note that this result is consistent with \eqref{eq:lim_c_large} (with $\delta = R_0$), though the two are not directly comparable due to the different order of limits.

This result again improves on that of \cite[Thm.~3]{Weinberger}, whose exponent for $N = \omega(M)$ turns out to be suboptimal.  This is analogous to what we showed in Figure \ref{fig:poisson_vs_multinomial}, so we omit the details here; the main point that we highlight is that we attain an exact error exponent while having a somewhat simpler analysis.

In the remainder of this section, we prove Theorem \ref{thm:superlinear}.
    
\subsection{Converse Bound}
    We fix $\delta < R_0$, and follow the same analysis as in Section \ref{sec:conv_isit}. In particular, equations \eqref{eq:pe_w} and \eqref{eq:pe_w_bound} still hold without change; from \eqref{eq:pe_w_bound}, we have
    \begin{equation}
        \mathbb{P}(m\in W) \leq \exp((\alpha-1)(\delta-R_0) M \log M) \cdot \mathcal{O}(1)^M = o(1),
    \end{equation}
    and \eqref{eq:pe_w} gives the following lower bound on the error probability $P_e$ of an arbitrary code:
    \begin{equation}
        P_e \geq \mathbb{P}(m\notin W) p(N,M,\delta M) = (1-o(1))p(N,M,\delta M) 
    \end{equation}
    Observe that $p(N,M,\delta M) \geq \delta^N$, since the probability of sampling only the first $\delta M$ molecules is $\delta^N$.
    Therefore,
    \begin{equation}
        P_e \ge (1-o(1))\delta^N,
    \end{equation}
    which implies
    \begin{equation}
        -\frac1N \log P_e \leq \log \frac1\delta + o(1).
    \end{equation}
    Since this is true for all $\delta < R_0$, we conclude that
    \begin{equation}
        -\frac1N \log P_e \leq \log \frac1{R_0} + o(1).
    \end{equation}
    
    \subsection{Achievability Bound}

    Consider a decoder that selects a message $i$ such that, among the $N$ decoded molecules, the number that are inside $A_i$ (including multiplicity) is maximized.  Note that here we consider multiplicity unlike earlier, the reason being that repetitions are naturally more prevalent in the regime $N = \omega(M)$. 
    Our analysis centers around the following analog of Lemma \ref{lem:sufficient_decoding_success} giving sufficient conditions for success.
    
    \begin{lemma}
    Let $\epsilon \in (0,1)$ and $\eta \in (0,1)$ be fixed.  Call a molecule in $A_i$ \emph{undersampled} if it appears at most $\frac{\eta N}{M}$ times in the size-$N$ multiset of sampled molecules (before sequencing errors).  Then, decoding succeeds if the following conditions hold:
    \begin{enumerate}
        \item[(i)] At most $(1-R_0-3\epsilon) M$ molecules are undersampled;
        \item[(ii)] There are fewer than $\epsilon \eta N$ sequencing errors;
        \item[(iii)] The codebook satisfies $|A_i \cap A_j| < (R_0 + \epsilon) M$ for all $i\neq j$.
    \end{enumerate}
    \end{lemma}
    \begin{proof}
    Let $i$ be the true message, and let $j$ be any other message. Let $\tilde{S}_0$ be the size-$N$ multiset of molecules sampled before sequencing errors, and let $\tilde{S}$ be the size-$N$ multiset of molecules produced after sequencing errors. Property (iii) gives $|A_i \setminus A_j| > (1-R_0-\epsilon)M$, so by property (i), the set $A_i \setminus A_j$ must contain at least $2\epsilon M$ molecules in $\tilde{S}_0$ that are not undersampled.  By the definition of being undersampled, this implies that $\tilde{S}_0$ contains at least $2\epsilon\eta N$ molecules (including multiplicity) from $A_i \setminus A_j$.

    Next, note that each sequencing error decreases the number of molecules (including multiplicity) from $A_i\setminus A_j$ by at most 1, so by property (ii), $\tilde{S}$ must contain more than $\epsilon \eta N$ molecules from $A_i \setminus A_j$. 
    On the other hand, every molecule in $A_j\setminus A_i$ must arise from a sequencing error, so again using property (ii), $\tilde{S}$ must contain be fewer than $\epsilon \eta N$ molecules (including multiplicity) from $A_j \setminus A_i$.  

    While we framed our decoder as maximizing the number of molecules in $A_{(\cdot)}$ observed, this is clearly equivalent to minimizing the number of molecules outside $A_{(\cdot)}$ observed. It follows that $i$ is preferred over $j$, as desired.
    
    \end{proof}

    From Lemma \ref{lem:collision_free}, there exists a codebook of size $\exp( (\alpha-1)R_0 M \log M )$ with $|A_i \cap A_j| < (R_0 + \epsilon)M$ for all $i\neq j$, thus ensuring condition (iii) above.  This can be still be used here because the construction of this codebook does not depend on $N$.

    Bounding the probability that condition (ii) fails is essentially the same as \eqref{eq:p_super_exp}: The number of sequencing errors follows a binomial distribution, so the probability of seeing more than $\epsilon \eta N$ sequencing errors is at most 
    \begin{equation}
        2^N (o(1))^{\epsilon \eta N} = \exp(-\omega(N)),
    \end{equation}
    since each invocation of sequencing has $o(1)$ error probability.

    It remains to consider condition (i).  Let $B$ be any subset of the input molecules of size $(1-R_0-3\epsilon)M$, and let $Z$ be the total number of times we sample molecules from $B$. Then $Z$ follows a binomial distribution with $N$ trials and success rate $\frac{|B|}{M} = 1-R_0-3\epsilon$.  Hence, the Chernoff bound gives
    \begin{multline}
        \mathbb{P}(Z \leq (1-R_0-3\epsilon)\eta N) \\ \leq \exp(-N \, D((1-R_0-3\epsilon)\eta || (1-R_0-3\epsilon))).
    \end{multline}
    On the other hand, if all of the molecules in $B$ are undersampled, then we must have $Z \leq (1-R_0-3\epsilon)\eta N$ by the definition of being undersampled.  Hence, and taking a union bound over all possible $B$ (there are at most $2^M$ such choices), the probability of condition (i) occurring is at most
    \begin{equation}
        2^M \cdot \exp(-D((1-R_0-3\epsilon)\eta || (1-R_0-3\epsilon)) \cdot N).
    \end{equation}
    Combining the failure events (i) and (ii) gives that the error probability $P_e$ satisfies
    \begin{multline}
        P_e \leq 2^M \cdot \exp(-D((1-R_0-3\epsilon)\eta || (1-R_0-3\epsilon)) \cdot N) \\ +  \exp(-\omega(N)).
    \end{multline}
    Noting that the first term decays exponentially in $N$, while the second term decays to zero faster than exponential, we have $-\frac1N \log P_e \geq  D((1-R_0-3\epsilon)\eta || (1-R_0-3\epsilon)) + o(1)$, 
    where we used the fact that $M = o(N)$.
    
    In other words, for arbitrarily small $\epsilon>0$ and $\eta>0$, we have
    \begin{equation}
    \lim_{N\rightarrow \infty} -\frac1N \log P_e \geq D((1-R_0-3\epsilon)\eta || (1-R_0-3\epsilon)).
    \end{equation}
    Taking the infimum over $\epsilon$ and $\eta$, it follows that
    \begin{equation}
        \lim_{N\rightarrow \infty} -\frac1N \log P_e \ge D(0 \| 1-R_0) = \log \frac{1}{R_0}
    \end{equation}
    as required.

\section{Low-Rate Regime Without Sequencing Errors} \label{sec:lowrate_no_seq}

In Section \ref{sec:const_rate}, we established exact error exponents for the case of a constant rate; letting $J$ denote the number of messages, a constant rate corresponds to $J = e^{\Theta(M \log M)}$, or $\log J = \Theta(M \log M)$.  We showed in \eqref{eq:lim_delta_small} (with $\delta = R_0$) that as the rate approaches zero, the error exponent grows unbounded (albeit very slowly). This motivates the question of how the error probability behaves in \emph{low-rate regimes}, where the number of messages is instead grows as $e^{o(M \log M)}$, i.e., $\log J = o(M \log M)$.  Our earlier results suggest that the optimal error probability behaves as $P_e^*(M) = e^{-\omega(N)}$, and our goal now is to more precisely determine the scaling and constant factors in the exponent.  This regime is relevant to scenarios requiring ultra-high reliability at the expense of a lower rate.

The low-rate regime turns out to be significantly more delicate.  The results vary significantly depending on precisely on how $J$ scales with respect to $M$, and also depending on what model is adopted for sequencing errors.  To decouple the challenges around sampling errors and sequencing errors, we first focus (in this section) on the case that there are no sequencing errors, i.e., whenever a molecule is sampled, it is received perfectly.  In Section \ref{sec:lowrate_seq}, we will drop this assumption.

We note that results on the \emph{short-molecule regime} also involve having $\exp(o(M\log M))$ messages (e.g., see \cite[Sec.~7.3]{shomorony2022information} and\cite{gerzon}), but overall their goals are substantially different from ours; these works have focused on capacity bounds with $L < \log M$, whereas we are interested in error exponents while still maintaining $L > \log M$.

\subsection{Summary of Scaling Laws} \label{sec:scaling2}

Following the preceding discussion, we briefly summarize our notation and assumed scaling as follows for easier cross-referencing later:
\begin{itemize}
    \item There are $M$ molecules in input.
    \item The inner code contains $M^{\alpha}$ distinct molecules, where $\alpha > 1$ remains constant as $M \rightarrow \infty$.
    \item The molecule length $L = \beta \log M$ will not play a direct role here; intuitively, any impact of varying $\beta$ is fully captured by the parameter $\alpha$ in the previous dot point (recall that $\alpha = \beta R_{\rm in}$).
    \item The number of messages is denoted by $J$, and henceforth we assume that $\log J = o(M \log M)$.
    \item The number of samples is $N$, for which we only assume that $N = \Omega(M)$ (i.e., $N$ may be either linear or super-linear with respect to $M$).
    \item $P_e^*(M)$ refers to optimal error probability under all possible encoders and decoders, again subject to being in the concatenated coding based class described in Definition \ref{def:concat_class}.
\end{itemize}
We note that in absence of sequencing errors, inner coding is not actually necessary, and one could safely make use of all $|\mathcal{X}|^L$ length-$L$ sequences.  However, we maintain full generality (i.e., general $\alpha > 1$) here in anticipation of Section \ref{sec:lowrate_seq}, where there are sequencing errors and inner coding is required.

\subsection{Initial Non-Asymptotic Bounds}

We first introduce two useful definitions characterizing how ``well-separated'' certain collections of codewords can be.

\begin{dfn} \label{def:K1}
    {\em (Achievability Separation Parameter $K_1$)} Given $(M,J,\alpha)$, let $K_1$ be the smallest integer such that there exist codewords without repeated molecules (i.e. no multisets) $A_1, A_2, \ldots, A_J$ satisfying $|A_i \cap A_j | \le K_1$ for all $i\neq j$.  Here each $A_i$ is a size-$M$ subset of the size-$M^{\alpha}$ inner codebook.
\end{dfn}

\begin{dfn} \label{def:K2}
    {\em (Converse Separation Parameter $K_2$)} Given $(M,J,\alpha)$, let $K_2$ be the largest integer such that for all possible codewords $A_1, A_2, \ldots, A_{J/2}$ (allowing multisets), there exists $i\neq j$ such that $|A_i\cap A_j| \geq K_2$.
\end{dfn}

\noindent Note that $K_1$ and $K_2$ implicitly depend not only on $J$, but also on the size $M$ of each outer codeword, and the constant $\alpha$ such that the inner code has size $M^{\alpha}$.

\begin{theorem}
Under the preceding setup without sequencing errors, we have
\begin{equation}
    \frac14 \left(\frac{K_2}{M}\right)^N \leq P_e^*(M) \leq \binom{M}{K_1} \left(\frac{K_1}{M}\right)^N,
\end{equation}
and the upper bound can be attained even when multisets are disallowed.
\label{thm:main}
\end{theorem}
\begin{proof}
    \underline{Achievability bound}: We use $A_1, A_2, \ldots, A_J$ from Definition \ref{def:K1} as a codebook (thus ensuring that there are no multisets). Note that if we observe more than $K_1$ distinct molecules, there is a unique $A_i$ that contains all of them.  Thus, the decoding rule that searches for such an $A_i$ will succeed.  The error probability is bounded above by the probability of seeing at most $K_1$ molecules, which is bounded above by $\binom{M}{K_1} \left(\frac{K_1}{M}\right)^N$.

    \underline{Converse bound}:  We first show that for any messages $i$ and $j$ satisfying $|A_i \cap A_j| \geq K_2$, conditioned on the true message being $i$ or $j$, the error probability is at least $\frac12(\frac{K_2}{M})^N$. 
    We use the following genie argument:
    \begin{itemize}
    \item The encoder, upon receiving message $i$, writes the molecules in $A_i$ as usual.
    \item The molecules in $A_i\cap A_j$ are \emph{tagged}.  If $A_i$ and $A_j$ are multisets, then the number tagged is equal to the multiplicity in $A_i\cap A_j$. For example, if $x$ appears 3 times in $A_i$ and 2 times in $A_j$, then 2 copies of $x$ are tagged.
    \end{itemize}
    The decoder can always ignore the tags, so any lower bound in this setup is valid for the original setup.

We consider a ``bad event'' in which all molecules that the decoder receives are tagged molecules.  Let $(y_1, y_2,\ldots, y_N)$ be any such sequence of (tagged) molecules. 
Conditioned on the encoder receiving message $i$, the likelihood of seeing this sequence $(y_1, y_2,\ldots, y_N)$ is 
\begin{equation}
    \mathbb{P}(y_1, \ldots, y_N | i) = \prod_{r=1}^N \mathbb{P}(y_r | i) = \prod_{r=1}^N \frac{(\#y_r\text{ in }A_i \cap A_j)}{M}.
\end{equation}

If we do the same computation conditioned on message $j$ instead of $i$, we find that the likelihood is the same. Therefore, the decoder cannot do better than a random guess, giving a conditional error probability of at least $\frac{1}{2}$.  In addition, since we are considering $(i,j)$ satisfying $|A_i \cap A_j| \geq K_2$, the probability of only seeing tagged molecules is at least $(\frac{K_2}{M})^N$.  Combining these findings, we deduce that conditioned on $m \in \{i,j\}$, the error probability is at least $\frac12(\frac{K_2}{M})^N$.

We now move from considering a fixed pair $(i,j)$ to considering the entire codebook.  To do so, we define a \emph{collision pair} to be \emph{any} pair $(i,j)$ such that $|A_i \cap A_j| \geq K_2$.   
We claim that there exists $J/2$ distinct integers $i_1, i_2,\ldots, i_{J/4}$ and $j_1, j_2,\ldots, j_{J/4}$ such that $({i_k},{j_k})$ form a collision pair for each $1\leq k\leq J/4$.  This is seen as follows:
    \begin{itemize}
        \item Maintain a list of codewords, initialized to be the entire codebook $(A_1, A_2,\ldots, A_J)$.  In addition, maintain a collection of collision pairs, initially empty.
        \item As long as the list of codewords has size at least $J/2$, identify a collision pair among them (which is possible by the definition of $K_2$, see Definition \ref{def:K2}), remove these two codewords from the list of codewords, and add this pair to the list of collision pairs.
    \end{itemize}
    This procedure immediately gives the required collision pairs $({i_k},{j_k})$.

    For each collision pair indexed by $(i,j)$, the above-established conditional lower bound of $\frac12(\frac{K_2}{M})^N$ applies.  Hence, the overall error probability satisfies
    \begin{align}
        &P_e^*(M) \nonumber \\ &\geq \sum_{k=1}^{J/4} \mathbb{P}({\rm error}|m=i_k \lor m=j_k) \cdot \mathbb{P}(m=i_k \lor m=j_k)\\
        &\geq \frac J4 \cdot \frac12 \left(\frac{K_2}{M}\right)^N \frac2J = \frac14 \left(\frac{K_2}{M}\right)^N.
    \end{align}
\end{proof}

\subsection{Bounds on $K_1$ and $K_2$}

Having provided non-asymptotic bounds in terms of $K_1$ and $K_2$ from Definitions \ref{def:K1} and \ref{def:K2}, we now proceed to bound these quantities themselves.

\begin{theorem}
    For all $c'>0$, the quantity $K_1$ from Definition \ref{def:K1} is bounded as follows:
        \begin{equation}
            K_1 \leq \left\lceil \max\left(\frac{\log J}{c' \log M}, e \cdot M^{2-\alpha+c'}\right)\right\rceil. \label{eq:k1_bound} 
        \end{equation}
    \label{thm:k1_bound}
    \end{theorem}
    \begin{proof}
    We temporarily let $K$ denote the right-hand side of \eqref{eq:k1_bound}:
    \begin{equation}
        K = \left\lceil \max\left(\frac{\log J}{c' \log M}, e \cdot M^{2-\alpha+c'}\right)\right\rceil. \label{eq:K_def}
    \end{equation}
    By the definition of $K_1$, it suffices to show that there exists a codebook such that $|A_i \cap A_j| \leq K$ for all $(i,j)$.
    
    Similarly to the proof of Lemma \ref{lem:collision_free}, we use a greedy argument with an index-based codebook.  We sequentially choose $A_1,A_2,\dotsc$ arbitrarily, subject to avoiding choices that are ``blocked'' in the sense of having intersection exceeding $K$ with a previously-selected codeword.  Whenever a new codeword is chosen, the number of additional codewords that become blocked is upper bounded by the following (analogous to \eqref{eq:num_blocked}):
    \begin{equation}
        \binom{M}{K} (M^{\alpha-1})^{M-K} \leq \left(\frac{eM}{K}\right)^{K} (M^{\alpha-1})^{M-K}.
    \end{equation}
    Moreover, the total number of choices for index-based codewords is $(M^{\alpha-1})^M$. As a result, we can continue the greedy method above for at least the following number of iterations (analogous to \eqref{eq:num_chosen}):
    \begin{align}
        \frac{(M^{\alpha-1})^M}{\left(\frac{eM}{K}\right)^{K} (M^{\alpha-1})^{M-K}} = \frac{(M^{\alpha-1})^{K}}{\left(\frac{eM}{K}\right)^{K}} &= \left(\frac{M^{\alpha-2} \cdot K}{e}\right)^K \nonumber \\ &\geq (M^{c'})^K \geq J,
    \end{align}
    where the last two inequalities hold since $K$ is at least as large as each term in the $\max$ in \eqref{eq:K_def}. 
    We conclude that there exist $J$ codewords $A_1, \ldots, A_J$ such that $|A_i \cap A_j| \leq K$.
    \end{proof}
    
    \begin{theorem}
        For any $J \le 2M^{\alpha M}$, the quantity $K_2$ from Definition \ref{def:K1} is bounded as follows:
        \begin{equation}
            K_2 \geq \Big\lfloor \frac{1}{\alpha} \frac{\log (J/2)}{\log M} \Big\rfloor \label{eq:K2_floor}
        \end{equation}
        \label{thm:k2_bound}
    \end{theorem}
    \begin{proof}
        We temporarily define  
            $K = \big\lfloor \frac{1}{\alpha} \frac{\log (J/2)}{\log M}\big\rfloor$ 
        to denote the right-hand side of \eqref{eq:K2_floor}.  
        By the assumption $J \le 2M^{\alpha M}$, we see that $K$ is a non-negative integer with value less than $M$.
        
        Observe that the number of multisets of $\{1,2,\ldots, M^{\alpha}\}$ of size $K$ is upper bounded by $(M^{\alpha})^K = M^{\alpha K}$, which is at most $J/2$ by the definition of $K$.  It follows that given  multisets $A_1, A_2, \ldots, A_{J/2}$, if we let $A'_i$ be an arbitrary subset of $A_i$ of size $K$ (say, the first $K$ molecules in some pre-specified order) for each $i \le J/2$, then $A'_1, A'_2, \ldots, A'_{J/2}$ cannot be all distinct.  Thus, there must exist some $i,j$ with $A'_i = A'_j$, which implies that $|A_i \cap A_j| \geq K$, and thus $K_2 \geq K$. 
    \end{proof}

    \subsection{Discussion on Scaling Regimes of $J$}

    So far, we have not assumed any specific scaling law; the preceding results hold for all $(M, \alpha, J, N)$. In the following subsections, we will incorporate the scaling laws from Section \ref{sec:scaling2} to derive asymptotic estimates for $K_1$ and $K_2$ as $M \to \infty$ (with $J$ and $N$ depending on $M$), which will in turn give asymptotic results on the error exponent.

    The analysis and results turn out to differ depending on whether the number of messages $J$ is above or below a threshold of roughly $\exp(M^{2-\alpha})$.  One reason for this is that when $J$ gets small enough, forcing all $M$ molecules to be distinct becomes suboptimal.  As an extreme example, when $J = 2$, a natural strategy is to let all $M$ input molecules be identical, and choose between one of two such molecules depending on the message.  The specific threshold $\exp(M^{2-\alpha})$ arises because the relevant $K_i$ values ``saturate'' to roughly $M^{2-\alpha}$ when $J$ decreases below $\exp(M^{2-\alpha})$; that is, the second term in Theorem \ref{thm:k1_bound} becomes dominant, and the bound in Theorem \ref{thm:k2_bound} becomes loose (we will give an improved counterpart for the small-$J$ regime in Theorem \ref{thm:k_large} below).

    We handle the case $J > \exp(M^{2-\alpha})$ in Section \ref{sec:large_J}, and the case $J \le \exp(M^{2-\alpha})$ in Section \ref{sec:small_J}.

\subsection{Error Exponents for $J > \exp(M^{2-\alpha})$} \label{sec:large_J}

    Here we consider the case $J > \exp(M^{2-\alpha})$; note that we allow $\alpha > 2$, in which case the condition $J > M^{2-\alpha}$ is automatically satisfied.

\begin{theorem}
    Consider the scaling regime described in Section \ref{sec:scaling2}. 
    Suppose that there exists $c>0$ such that $J \geq \exp(M^{2-\alpha+c})$, and that $\frac{\log J}{\log M} \rightarrow \infty$ and $\log J = o(M \log M)$. Then, in the absence of sequencing errors, we have
    \begin{equation}
        \lim_{M \to \infty} \frac{\log \frac{1}{P_e^*(M)}}{N \log {\frac{M \log M}{\log J}}} = 1.
    \end{equation}
\label{thm:main_exponent}
\end{theorem}
\begin{proof}
We use \eqref{eq:k1_bound} with $c/2$ replacing $c$.  
Observe that the assumption  $J \geq \exp(M^{2-\alpha+c})$ gives
\begin{equation}
    e \cdot M^{2-\alpha+{c/2}} = o\left(\frac{M^{2-\alpha+c}}{\log M}\right) = o\Big(\frac{\log J}{\log M}\Big).
\end{equation}
Therefore, we obtain from \eqref{eq:k1_bound} that
\begin{equation}
K_1 \leq \left\lceil\max\left(\frac{\log J}{(c/2) \log M}, e \cdot M^{2-\alpha+{c/2}}\right)\right\rceil = \mathcal{O}\Big(\frac{\log J}{\log M}\Big), \label{eq:K1_scaling}
\end{equation}
which implies
\begin{equation}
    \log \frac{M}{K_1} 
    \geq \log \frac{M\log M}{\log J} + \mathcal{O}(1) = (1+o(1)) \log \frac{M\log M}{\log J},  \label{eq:fracMK1}
\end{equation}
where the last step uses the assumption $\frac{M\log M}{\log J} \to \infty$.

We now bound the error probability using Theorem \ref{thm:main}:
\begin{align}
    \log \frac{1}{P_e^*(M)} &\geq -\log \binom{M}{K_1} + N \log \frac{M}{K_1} \\
    &\geq -M + N (1+o(1)) \log \frac{M\log M}{\log J}\\
    & = N \log\left(\frac{M\log M}{\log J}\right) (1 + o(1)),
\end{align}
where the second step uses \eqref{eq:fracMK1} and $\binom{M}{K_1} \le 2^M \le e^M$, and the last step uses $N = \Omega(M)$ and $\frac{M\log M}{\log J} \to \infty$.

Similarly, since $K_2 = \Omega(\frac{\log J}{\log M})$ (see \eqref{eq:K2_floor}), an analogous argument using Theorem \ref{thm:main} gives
\begin{align}
    \log \frac{1}{P_e^*(M)}& \leq N \log \frac{M}{K_2} + \log 4  \nonumber \\
    &\leq N \log\left(\frac{M\log M}{\log J}\right) (1 + o(1)).
\end{align}
Combining these bounds gives the desired result.
\end{proof}

\begin{cor}
    If $J = \exp(M^s)$ for some $s$ satisfying $\max(0, 2-\alpha) < s < 1$, then
    \begin{equation}
        \lim_{M \to \infty} \frac{\log \frac{1}{P_e^*(M)}}{N \log M} = 1-s.
    \end{equation}
    \label{cor:exp_m_s}
\end{cor}
\begin{proof}
    Under the assumed scaling, we have $\log\frac{M \log M}{\log J} = \log\frac{M}{M^s} (1+o(1)) = \big((1-s) \log M\big)(1+o(1))$, so the result follows from Theorem \ref{thm:main_exponent}.
\end{proof}

\subsection{The Case $J \leq \exp(M^{2-\alpha})$ with $\alpha \in (1,2)$} \label{sec:small_J}

In the case that $\alpha \in (1,2)$ and $J \leq \exp(M^{2-\alpha})$, the situation becomes more subtle, and as we hinted above, it becomes beneficial to allow repeated molecules in the outer codewords (i.e., multisets).  Note that having repeated molecules precludes being index-based according to Definition \ref{dfn:index_based}, but will still essentially use the same idea by taking a ``smaller'' index-based code and performing trivial repetition.

To understand the difference between the cases of multisets and no multisets, we proceed to study the two separately.

\subsubsection{The Setting Without Multisets}

In the following, we make the very mild assumption $J > 2M^\alpha$ (i.e., the number of messages at least exceeds twice the inner code size). 

\begin{theorem}
    Let $K_3$ be defined similarly as $K_2$ (Definition \ref{def:K2}), except that multisets are now disallowed. If $J > 2M^\alpha$ with $\alpha \in (1,2)$, then
    \begin{equation}
        K_3 \geq M^{2-\alpha} - \frac{M}{J/2 -1}.
    \end{equation}
    \label{thm:k_large}
\end{theorem}
Before presenting the proof, let us start with an intuitive argument. If two sets of size $M$ are picked uniformly among $\{1,2,\ldots, M^{\alpha}\}$, then the expected number of collisions is given by $M^{2-\alpha}$. Intuitively, Theorem \ref{thm:k_large} shows that the optimal choice is not much better than simply choosing sets randomly. This intuitive argument is not used in the proof, but provides a hint as to why $2-\alpha$ is the correct exponent.

\begin{proof}
    The argument is analogous to the classical Plotkin bound \cite[Sec.~5.8]{gallager1968information}, but we provide the full details for completeness.
    To simplify notation, let $J' = J/2$ be the number of codewords in the definition of $K_2$; these codewords, represented as sets of molecules, are denoted by $A_1, A_2,\ldots, A_{J'}$.  We further represent these sets using 0-1 vectors of length $M^{\alpha}$, denoted by $v_1, v_2, \ldots, v_{J'}$.  For each $v_i$, the $\ell$-th coordinate is 1 if and only if $A_i$ contains the $\ell$-th molecule of the inner code. Observe the size of the intersection $|A_i \cap A_j|$ is equal to the inner product $v_i \cdot v_j$.
    
    By construction, we have $||v_i||_1 = M$ and $||v_i||_2^2 = M$.  Since the $\ell_1$-norm is simply the sum of entries for non-negative vectors, we have
    \begin{equation}
        \Big\|\sum_i v_i\Big\|_1 = M \cdot J'.
    \end{equation}
    By the inequality between $\ell_1$-norm and $\ell_2$-norm (via Cauchy-Schwarz inequality), we have
    \begin{equation}
        \Big\|\sum_i v_i\Big\|_2^2 \geq \frac{(M\cdot J')^2}{M^\alpha} = M^{2-\alpha} (J')^2,
    \end{equation}
    and hence
    \begin{align}
    M^{2-\alpha} (J')^2 \leq \Big\|\sum_i v_i\Big\|_2^2 &= \sum_i ||v_i||_2^2 + \sum_{i \neq j} v_i \cdot v_j \\ &= M \cdot J' + \sum_{i \neq j} v_i \cdot v_j.
    \end{align}
    Rearranging gives
    \begin{equation}
       \sum_{i \neq j} v_i \cdot v_j \leq  M^{2-\alpha} (J')^2 - M \cdot J'.
    \end{equation}
    Then, there exists some $(i,j)$ such that $v_i \cdot v_j$ is at least as high as the average:
    \begin{equation}
        v_i \cdot v_j \ge \frac{1}{J'(J'-1)}(M^{2-\alpha} (J')^2 - M \cdot J') \geq M^{2-\alpha}-\frac{M}{J'-1},
    \end{equation}
    Therefore, there exists a pair $(i,j)$ such that $|A_i \cap A_j| \geq M^{2-\alpha}-\frac{M}{J'-1}$.
\end{proof}

This leads to the following corollary.  

\begin{cor} \label{cor:no_multi}
    Consider the scaling regime described in Section \ref{sec:scaling2}. 
    If multiset codewords are not allowed, and $J$ satisfies $J > 2M^\alpha$ and $J \leq \exp(M^{2-\alpha})$ with $\alpha \in (1,2)$, then
    \begin{equation}
        \lim_{M \to \infty} \frac{\log \frac{1}{P_e^*(M)}}{N\log M} = \alpha-1.
    \end{equation}
\end{cor}
\begin{proof}
    By the converse part of Theorem \ref{thm:main}, with $K_2$ replaced by $K_3$ due to disallowing multisets (upon which the proof holds verbatim), we have
    \begin{equation}
        P_e^*(M) \geq \frac14 \left(\frac{K_3}{M}\right)^N.
    \end{equation}
    From Theorem \ref{thm:k_large} and the fact that $J = \omega(M)$ (since $J > 2M^{\alpha}$), we have $K_3 \ge M^{2-\alpha} - o(1)$, and hence
    \begin{equation}
        \frac{\log \frac{1}{P_e^*(M)}}{(\alpha-1) N\log M}\leq 1+o(1).
    \end{equation}
    For the achievability part, we use Theorem \ref{thm:k1_bound}, in which multisets are automatically disallowed by Definition \ref{def:K1}.  We substitute $c'=\frac{1}{\log M}$ (since Theorem \ref{thm:k1_bound} is non-asymptotic, $c'$ is allowed to depend on $M$) to obtain
    \begin{equation}
            K_1 \leq \left\lceil\max\left(\log J, e \cdot M^{2-\alpha} \cdot M^{1/{\log M}}\right)\right\rceil = \mathcal{O}(M^{2-\alpha}). \label{eq:K1_bigO}
    \end{equation}
    Therefore,
    \begin{equation}
        \log \frac{1}{P_e^*(M)} \geq N \log \frac{M}{K_1} - M \geq (\alpha-1)\big(N \log M\big) (1+o(1)),
    \end{equation}
    where the first inequality uses the upper bound in Theorem \ref{thm:main} along with $\binom{M}{K_1} \le 2^M \le e^{M}$, and the second inequality uses \eqref{eq:K1_bigO}.  This achievability bound matches the above converse, and the proof is complete.
    
\end{proof}

\subsubsection{The Setting with Multisets Allowed}

Next, we state the analog of Corollary \ref{cor:no_multi} for the case that multisets are allowed.

\begin{theorem} \label{thm:multi_allowed}
    Consider the scaling regime described in Section \ref{sec:scaling2}, and assume that $\alpha \in (1,2)$. 
    Suppose that multiset codewords are allowed, and let $J = \exp(M^s)$ for some $s \in (0,2-\alpha)$.  Then,
    \begin{equation}
        \lim_{M \to \infty} \frac{\log \frac1{P_e^*(M)}}{N \log M} = \frac{\alpha-s}2.
    \end{equation}
\end{theorem}

Observe that when $s < 2-\alpha$, it holds that $\alpha-1 < \frac{\alpha-s}2$, showing that the error exponent is indeed strictly higher than that of Corollary \ref{cor:no_multi}.  Naturally, this gap diminishes when we take $s$ increasingly close to $2-\alpha$.

To prove Theorem \ref{thm:multi_allowed}, we proceed by presenting the achievability part and then a matching converse.

\subsubsection{Achievability proof for Theorem \ref{thm:multi_allowed}}

We fix $t \in (0,1)$ and consider an encoder that only chooses $M^t$ molecules instead of $M$ molecules, but it repeats each of them $M^{1-t}$ times (for a total of $M$).  Observe that since sampling is done with replacement, this is precisely equivalent to only having $M^t$ molecules in the first place, and only writing them once each.\footnote{Note that if the problem formulation allowed using $M$ input molecules \emph{or fewer} instead of \emph{exactly} $M$, then we could simply use these $M^t$ molecules and avoid multisets.}

Accordingly, we define $M' = M^t$ and $\alpha ' = \alpha/t$ so that the encoder chooses $M'$ molecules and the total number of available molecules (i.e., the inner codebook size) is $M^\alpha = (M')^{\alpha'}$.  
Since $J = \exp(M^s)$, setting $s' = s/t$ gives $J = \exp((M')^{s'})$. If $\max(2-\alpha', 0) < s' < 1$, then we can substitute $(M', \alpha', s')$ for $(M,\alpha, s)$ in Corollary \ref{cor:exp_m_s} (due to the above-mentioned equivalence) to obtain
\begin{equation}
    \log \frac1{P_e^*(M)} \le (1-s)(N \log M')(1+o(1)),
\end{equation}
and therefore
\begin{align}
    \log \frac1{P_e^*(M)} &\le (1-s') \cdot t \cdot (N \log M) \cdot (1+o(1)) \nonumber \\  &= (t-s) \cdot (N \log M) \cdot (1+o(1)).
\end{align}
Since $s>0$, we have $s'>0$. To obtain the best error exponent, we want to maximize $t$ while maintaining the condition $2-\alpha' < s' < 1$, which is equivalent to $t < \frac{s+\alpha}2$ and $t > s$. We can make $t$ arbitrarily close to $\frac{s+\alpha}2$, so that the limiting value of $\frac{\log \frac{1}{P_e^*(M)}}{N \log M}$ can be made arbitrarily close to $\frac{s+\alpha}2 -s = \frac{\alpha-s}2$.

\subsubsection{Converse proof for Theorem \ref{thm:multi_allowed}}

To prove the converse part of Theorem \ref{thm:multi_allowed}, we first provide a lower bound on $K_2$ from Definition \ref{def:K2}.

\begin{theorem}
Under the choice $J = \exp(M^s)$ with $0 < s \leq 2-\alpha$, when $M$ is large enough for the inequality $J > 2M^{\alpha}(\log_2 M +1)$ to hold, we have
\begin{equation}
    K_2 \geq \min\left(\frac{M^{1+(s-\alpha)/2}}{(\log_2 M+1)^2} - \frac{M}{J/2-1}, \frac1{2\alpha} \cdot \frac{M^{1+(s-\alpha)/2}}{\log_2 M+1} \right).
\end{equation}
\label{thm:large_intersection_multiset}
\end{theorem}

\begin{proof}
As before, let $J' = J/2$ be the number of codewords in the definition of $K_2$, and let $A_1,A_2\dotsc,A_{J'}$ denote these codewords represented as multisets of molecules.  For each $A_i$, let $v_i$ be its length-$M^{\alpha}$ frequency vector. That is $v_i$ is equal to the number of occurrences of the $i$-th molecule (in the inner codebook) in $A_i$, and since multisets are allowed, we may have $v_i > 1$.

    For each integer $i \in [1,J']$ and $\ell \in [0,\log_2 M]$, construct new vectors $v_{i,\ell}$ such that for each entry of $v_i$ (taking a value in $\{0,1,\dotsc,M\}$), the corresponding entry of $v_{i,\ell}$ equals the $\ell$-th bit in its binary expansion (with $\ell = 0$ corresponding to the least significant bit).
    By summing the contributions of these coordinates, we have
    \begin{equation}
        v_i = \sum_{\ell=0}^{\log_2 M} 2^\ell v_{i,\ell}.
    \end{equation}
    Observe that we still have $||v_i||_1 = M$ for all $i$ (since $||v_i||_1$ simply adds the multiplicities of all molecules), and therefore $|| \sum_i v_i||_1 = M \cdot J'$, and
    \begin{align}
        M \cdot J' = \Big\| \sum_i v_i\Big\|_1 &= \Big\| \sum_{\ell} \sum_{i} 2^\ell v_{i,\ell} \Big\|_1 \\ &= \sum_{\ell=0}^{\log_2 M} 2^\ell \Big\|\sum_i v_{i,\ell}\Big\|_1.
    \end{align}
    Hence, there exists $\ell$ such that
    \begin{equation}
        \Big\|\sum_{i} v_{i,\ell}\Big\|_1 \geq \frac{M\cdot J'}{(\log_2 M+1) \cdot 2^\ell}.
        \label{eq:exists_l}
    \end{equation}
    We now consider two cases separately.

    \underline{Case 1} ($2^\ell < M^{1-(s+\alpha)/2}$): 
    By the inequality relation of $\ell_1$-norm and $\ell_2$-norm, we have    
    \begin{equation}
        \Big\|\sum_{i} v_{i,\ell}\Big\|_2^2 \geq \frac{1}{M^\alpha} \Big\|\sum_{i} v_{i,\ell}\Big\|_1^2.
        \label{eq:cauchy_schwartz}
    \end{equation}
   Moreover, since $v_{i,\ell}$ is a 0-1 vector (because we simply extracted binary digits), $||v_{i,\ell}||_2^2 = ||v_{i,\ell}||_1$, so that
    \begin{equation}
        \sum_i \|v_{i,\ell}\|_2^2 = \Big\|\sum_{i} v_{i,\ell}\Big\|_1. \label{eq:L2eqL1}
    \end{equation}
    Hence, expanding the square in \eqref{eq:cauchy_schwartz} gives
    \begin{align}
    \frac{1}{M^\alpha} \Big\|\sum_i v_{i,\ell}\Big\|_1^2 &\leq \Big\|\sum_i v_{i,\ell}\Big\|_2^2 \\ &= \sum_i \|v_{i,\ell}\|_2^2 + \sum_{i \neq j} v_{i,\ell} \cdot v_{j,\ell} \\
    & \stackrel{\eqref{eq:L2eqL1}}{=} \sum_i \|v_{i,\ell}\|_1 + \sum_{i \neq j} v_{i,\ell} \cdot v_{j,\ell}.
    \end{align}
    Rearranging, we obtain
    \begin{align}
        \sum_{i \neq j} v_{i,\ell} \cdot v_{j,\ell}  &\geq \frac{1}{M^\alpha} \Big\|\sum_i v_{i,\ell}\Big\|_1^2  - \sum_i \|v_{i,\ell}\|_1 \\ &\stackrel{\eqref{eq:exists_l}}{\geq} \frac{M^{2-\alpha} (J')^2}{((\log_2 M+1) \cdot 2^\ell)^2} - M\cdot J'.
    \end{align}
 Since the maximum (over $J'(J'-1)$ ordered choices of $(i,j)$) is at least as high as the average, we conclude that there exists $i,j$ such that
 \begin{equation}
     v_{i,\ell} \cdot v_{j,\ell} \geq \frac{M^{2-\alpha}}{((\log_2 M+1) \cdot 2^\ell)^2} - \frac{M}{J'-1}
     \label{eq:exists_i_j_l}
 \end{equation}
We now interpret this statement in terms of intersections of outer codewords.  
Let $i,j,\ell$ satisfy \eqref{eq:exists_i_j_l}, and let $\tilde{X}$ be the set of all molecules $x$ such that $(v_{i,\ell})_x = (v_{j,\ell})_x = 1$.  We have shown that this value of $\ell \in [0,\log_2 M]$ and the corresponding codewords $A_i, A_j$ satisfy the following:
 \begin{itemize}
     \item $|\tilde{X}| \geq \frac{M^{2-\alpha}}{((\log_2 M+1) \cdot 2^\ell)^2} - \frac{M}{J'-1}$;
     \item For each $x \in \tilde{X}$, we have that $x$ appears at least $2^\ell$ times in $|A_i\cap A_j|$ (since the $\ell$-th binary digits of ($\text{\#$x$ in }A_i$) and ($\text{\#$x$ in }A_j$) are both 1).
 \end{itemize}
Hence, and recalling that $2^\ell < M^{1-(s+\alpha)/2}$ in the current case 1, we have
 \begin{align}
     |A_i \cap A_j| \geq |\tilde{X}| \cdot 2^\ell &\geq \frac{M^{2-\alpha}}{(\log_2 M+1)^2 \cdot 2^\ell} - \frac{M}{J'-1} \\ &\geq \frac{M^{1+(s-\alpha)/2}}{(\log_2 M+1)^2} - \frac{M}{J'-1},
 \end{align}
 which completes the proof for case 1.

\underline{Case 2} ($2^\ell \geq M^{1-(s+\alpha)/2}$): 
For each $i$, let $(B_i)_{i=1}^{J'}$ be sets such that $x \in B_i \Leftrightarrow (\#x\text{ in }A_i) \geq 2^\ell$.
Observe that for each $x$ such that $(v_{i,\ell})_x = 1$, we have $(\#x\text{ in }A_i) \geq 2^\ell$ and therefore $x \in B_i$.
We then obtain from \eqref{eq:exists_l} that
\begin{equation}
    \sum_i |B_i| \geq \sum_i ||v_{i,\ell}||_1 \geq \frac{M \cdot J'}{(\log_2 M+1) \cdot 2^\ell}. \label{eq:sum_abs_Bi}
\end{equation}
We now define
\begin{equation}
    K' = \Bigg\lceil\frac1{2\alpha} \cdot \frac{M^{1+(s-\alpha)/2}}{(\log_2 M+1) \cdot 2^\ell}\Bigg\rceil
\end{equation}
and claim that there exists $i,j$ with $|B_i \cap B_j| \geq K'$ via two sub-cases:
\begin{itemize}
    \item \underline{Case 2a} ($K'=1$): Further bounding \eqref{eq:sum_abs_Bi} via $2^{\ell} \le M$, we have
    \begin{equation}
        \sum_i |B_i| \geq \frac{J'}{\log_2M+1} > M^\alpha \geq |\cup_i B_i|, \label{eq:Bi_ineq}
    \end{equation}
    where the strict inequality follows since we have assumed $J' > M^{\alpha}(\log_2 M +1)$ in the theorem statement, and the final inequality holds because $\cup_i B_i$ is a subset of the set of all inner codewords (of which there are $M^{\alpha}$).  It follows from \eqref{eq:Bi_ineq} that the collection $(B_i)_{i=1}^{J'}$ cannot be disjoint. Thus there exists a pair $(i,j)$ with $|B_i \cap B_j| \geq 1 = K'$.
    
    \item \underline{Case 2b} ($K' \geq 2$): In this case, it is useful to define the following non-rounded version of $K'$:
    \begin{equation}
        \kappa = \frac1{2\alpha} \cdot \frac{M^{1+(s-\alpha)/2}}{(\log_2 M+1) \cdot 2^\ell},
    \end{equation}
    so that $K' = \lceil \kappa \rceil$.  The assumption $K' \ge 2$ implies that $\kappa > 1$ and thus $K' = \lceil \kappa \rceil \leq \kappa+1 \leq 2\kappa$. Therefore, 
    \begin{equation}
        K' \leq \frac1\alpha \cdot \frac{M^{1+(s-\alpha)/2}}{(\log_2 M+1) \cdot 2^\ell}.  \label{eq:K_case2b}
    \end{equation}
    Observe that re-arranging this equation gives
    \begin{equation}
        \frac{M}{(\log_2 M + 1) \cdot 2^{\ell}} \geq K' \alpha M^{(\alpha - s)/2} \ge K', \label{eq:K'_LB}
    \end{equation}
    where the last step holds because $\alpha > 1$ and $\alpha - s \ge \alpha - (2-\alpha) = 2(\alpha - 1) > 0$.
    
    We now proceed as follows:
    \begin{equation}
        \sum_i |B_i| \stackrel{\eqref{eq:sum_abs_Bi}}{\geq} \frac{M \cdot J'}{(\log_2 M+1) \cdot 2^\ell}
        \stackrel{\eqref{eq:K'_LB}}{\geq} J' \cdot K',
    \end{equation}
    and therefore
    \begin{equation}
        \sum_i (|B_i|-K'+1) \geq J'  \cdot K' - J' \cdot K' + J' = J'. \label{eq:J_Case2b_1}
    \end{equation}
    Suppose for contradiction that $|B_i \cap B_j| < K'$ for all $i,j$.  We claim that each $B_i$ has 
    \begin{equation}
        \binom{|B_i|}{K'} \geq |B_i|-K'+1 \label{eq:J_Case2b_2}
    \end{equation}
    subsets of size $K'$, which is seen via two cases:
     \begin{itemize}
        \item If $|B_i|<K'$, then the right-hand side is negative and the left-hand side is zero;
        \item If $|B_i|\geq K'$, then we can pick the first $K'-1$ elements and still have $|B_i|-K'+1$ choices for the last.
     \end{itemize}
    Then, we have
    \begin{equation}
        J' \stackrel{\eqref{eq:J_Case2b_1}}{\leq} \sum_i (|B_i|-K'+1) \stackrel{\eqref{eq:J_Case2b_2}}{\leq} \sum_i \binom{|B_i|}{K'} \leq \binom{M^\alpha}{K'}, \label{eq:J_upper_bound}
    \end{equation}
    where the last step uses the assumption $|B_i \cap B_j| < K'$, which implies that all of the $\binom{|B_i|}{K'}$ terms are counting \emph{distinct} size-$K'$ subsets of $\{1,2,\ldots, M^\alpha\}$.
    
    Next, since $2^\ell \geq M^{1-(s+\alpha)/2}$ (which we assumed for case 2), we can upper bound \eqref{eq:K_case2b} as follows:
    \begin{equation}
        K' \leq \frac{M^s}{\alpha (\log_2M+1)} \leq \frac{M^s}{\alpha \log_2M},
    \end{equation}
    which further implies $(M^\alpha)^{K'} = 2^{K' \alpha \log_2 M} \leq 2^{M^s}$.  
    Combining this with \eqref{eq:J_upper_bound} and recall that $K'\geq 2$, we obtain
    \begin{equation}
        J' \leq \binom{M^\alpha}{K'} \leq \frac{(M^\alpha)^{K'}}{(K')!}  \leq \frac{2^{M^s}}{2}.
    \end{equation}
    This contradicts the fact that $J' = \exp(M^s)/2$, which completes the proof by contradiction that $|B_i \cap B_j| \ge K'$ for some $(i,j)$.
\end{itemize}
Having established the above, let $(i,j)$ be a pair satisfying $|B_i \cap B_j| \geq K'$. 

By the definition of $B_i$, each element in $B_i$ must appear at least $2^\ell$ times in $A_i$. Therefore,
\begin{align}
    |A_i \cap A_j| \geq 2^\ell \cdot K' &= 2^\ell \cdot \Bigg\lceil\frac1{2\alpha} \cdot \frac{M^{1+(s-\alpha)/2}}{(\log_2 M+1) \cdot 2^\ell}\Bigg\rceil \\ &\geq \frac1{2\alpha} \cdot \frac{M^{1+(s-\alpha)/2}}{\log_2 M + 1}.
\end{align}
This completes the proof of case 2, and thus the proof of Theorem \ref{thm:large_intersection_multiset}.

\end{proof}

We now proceed to complete the proof of the converse part of Theorem \ref{thm:multi_allowed}.  
By Theorem \ref{thm:large_intersection_multiset}, we have
\begin{equation}
    \frac{K_2}{M} \geq \Omega \left(\frac{M^{(s-\alpha)/2}}{\log^2 M}\right),
\end{equation}
which implies
\begin{equation}
    \log \frac{K_2}{M} \geq \frac{s-\alpha}{2} \log M - \mathcal{O}(\log \log M).
    \label{eq:large_intersection_multiset_asymp}
\end{equation}
Now, Theorem \ref{thm:main} gives $P_e^*(M) \geq \frac14 \left(\frac{K_2}{M}\right)^N$, and taking the log and substituting \eqref{eq:large_intersection_multiset_asymp} gives
\begin{equation}
    \log P_e^*(M) \geq N \log \frac{K_2}{M} - \log 4 \geq \Big( N\frac{s-\alpha}{2} \log M \Big)(1+o(1)).
\end{equation}
Thus, we get the desired limiting behavior $\frac{\log \frac1{P_e^*(M)}}{N \log M} \le \big(\frac{\alpha-s}{2}\big)(1+o(1))$.

\section{Low-Rate Regime With Sequencing Errors} \label{sec:lowrate_seq}

For low-rate regimes in the presence sequencing errors, the error exponent depends on finer properties of the inner code, rather than only the assumption that it has $o(1)$ inner error probability.  Roughly speaking, this is because sequencing errors are only significant when they impact a constant fraction of sampled molecules, and this occurs with probability $p^{\Theta(N)}$, where $p = o(1)$ is the sequencing error probability.  This probability is insignificant when the overall error probability is $e^{-\Theta(N)}$ (e.g., in Theorems \ref{thm:main_result} and \ref{thm:superlinear}), but can become significant in the low-rate regime where the overall error probability is $e^{-\omega(N)}$ (as discussed at the start of Section \ref{sec:lowrate_no_seq}).

The dependence on finer properties of the inner code is somewhat in tension with the fact that we would like to use it in a ``black-box'' manner.  We approach this problem by studying three simple models for how sequencing errors occur, one of which is a ``worst-case'' view and thus maintains the desired ``black-box'' property.  The other two are not necessarily realistic, but serve to give an indication of how much the exponents might improve when the worst-case view is dropped.  (See Section \ref{sec:compare_models} for the comparisons.)


In more detail, we suppose that every time a molecule is sequenced and decoded, it equals the original molecule with probability $1-p$, while the remaining probability $p$ is said to constitute a \emph{sequencing error}.  When a sequencing error occurs, we consider the following (separate) models for what happens:
\begin{itemize}
    \item \underline{Erasure}: The original molecule is erased from the decoder output. As a result, the decoder may receive less than $N$ molecules. 
    \item \underline{Adversarial}: Whenever a sequencing error occurs, the decoded molecule is completely arbitrary, and we are interested in the worst-case error probability of the outer code with respect to such errors.  Stated differently, we view the incorrectly decoded molecules as being chosen by an adversary that has complete knowledge of the encoder, decoder, and message.
    \item \underline{Random}: The true molecule is replaced by a molecule chosen uniformly at random among all $M^\alpha$ molecules in the inner code, independently of the original molecule.
\end{itemize}
Note that $p$ is essentially the error probability of the inner code, though strictly speaking it is only an upper bound (e.g., in the random model, the inner code's error probability would more precisely be $p\big(1-\frac{1}{M^{\alpha}}\big)$).  For achievability results, the adversarial model is arguably the most desirable since it amounts to making no assumption on the details of the inner code (apart from its error probability).

Throughout the section, we adopt the natural assumption that $p \to 0$ as $M \to \infty$, i.e., a ``good'' inner code is used at an achievable inner rate (Definition \ref{def:concat_class}), and so its error probability is asymptotically vanishing.  In addition, while the regimes $J > \exp(M^{2-\alpha})$ and $J \le \exp(M^{2-\alpha})$ are both of interest (see Sections \ref{sec:large_J} and \ref{sec:small_J}), we focus our attention on the former.  This is because (i) we expect that the rate being ``low but not too low'' is of more interest, and (ii) the regime $J \le \exp(M^{2-\alpha})$ may become increasingly complicated, as it already introduced additional subtle issues and complications even without sequencing errors.


\subsection{Erasure model}

The analysis of the erasure model follows fairly simply from the analysis with no sequencing errors, since the latter was already based on the idea of counting how many molecules are ``lost''.

We first state a non-asymptotic bound, and then analyze its error exponent.  Recalling the definitions of $K_1$ and $K_2$ in Definitions \ref{def:K1} and \ref{def:K2}, we have the following generalization of Theorem \ref{thm:main}.

\begin{theorem} \label{thm:erasure_non_asymp}
    Under the erasure sequencing error model with sequencing error probability $p$, we have
    \begin{equation}
    \frac14 \max\left(p,\frac{K_2}{M}\right)^N \leq P_e^*(M) \leq \binom{M}{K_1} \left(p + \frac{K_1}{M}\right)^N.
    \end{equation}
    Moreover, the upper bound can be attained even when multisets are disallowed.
\end{theorem}
\begin{proof}
    For the achievability part, we use the same argument as in the proof of Theorem \ref{thm:main}; if we receive more than $K_1$ molecules, we are guaranteed to identify $A_i$ uniquely. The probability of seeing $K_1$ or fewer distinct molecules is now upper bounded by $\binom{M}{K_1} \left(p + \frac{K_1}{M}\right)^N$.  

    Similarly, for the converse part, we again use the genie argument from the proof of Theorem \ref{thm:main}.  The only difference is that the probability $\frac{K_2}{M}$ of seeing a tagged molecule is replaced by the probability of seeing a tagged molecule \emph{or} having an erasure.  Taking the maximum of the two associated probabilities gives $\max\big(p,\frac{K_2}{M}\big)$.
\end{proof}

With Theorem \ref{thm:erasure_non_asymp} in place, we can use the bounds on $K_1$ and $K_2$ established earlier to deduce the resulting error exponent.

\begin{cor}
    Consider the scaling regime described in Section \ref{sec:scaling2}.  
    If there exists $c>0$ such that $J \geq \exp(M^{2-\alpha+c})$, and it holds that $\frac{\log J}{\log M} \rightarrow \infty$ and $\frac{\log J}{M \log M} \to 0$, then under the erasure sequencing error model with sequencing error probability $p = o(1)$, we have
        \begin{equation}
        \lim_{M \to \infty}\frac{\log \frac1{P_e^*(M)}}{N \log \min( \frac1p, \frac{M\log M}{\log J})} = 1. \label{eq:erasure_exponent}
    \end{equation}
    \label{cor:erasure_exponent}
\end{cor}
\begin{proof}
    In the achievability part of the proof of Theorem \ref{thm:main_exponent} (which has the same assumptions on the scaling of $J$ as here), we established that $K_1 = \mathcal{O}\big( \frac{\log J}{\log M}\big)$. We therefore have
\begin{align}
    \log P_e^*(M) &\leq \log \binom{M}{K_1} + N \log \left(p + \frac{K_1}{M}\right) \\
    &\leq M + N \left(\log 2 + \max\left(\log p, \log \frac{K_1}{M}\right)\right)\\
    & =   N \cdot \max\left(\log p, \log \frac{\log J}{M\log M}\right) (1 + o(1)),
    \label{eq:erasure_exponent_ach}
\end{align}
where the last step uses the assumptions $p=o(1)$ and $\frac{\log J}{M \log M} \to 0$. 

For the converse part, in the proof of Theorem \ref{thm:main_exponent}, we already established that the following holds even when there are no sequencing errors:
\begin{equation}
    \log P_e^*(M) \geq N \log\left(\frac{\log J}{M\log M}\right) (1 + o(1))
\end{equation}
Since the probability of all molecules being erased is $p^N$ and the conditional error probability is trivially $1-o(1)$ when that occurs, we also have
\begin{equation}
    \log P_e^*(M) \geq N \log p - o(1),
\end{equation}
and combining the two lower bounds gives
\begin{equation}
    \log P_e^*(M) \geq N \max\left(\log p, \log\left(\frac{\log J}{M\log M}\right)\right) (1 + o(1)).
    \label{eq:erasure_exponent_conv}
\end{equation}
This completes the proof of Corollary \ref{cor:erasure_exponent}.
\end{proof}

\subsection{Adversarial model}

We now turn to the adversarial model, again starting with non-asymptotic upper and lower bounds on the optimal error probability.

\begin{theorem}
    Under the adversarial sequencing error model with sequencing error probability $p$, we have
    \begin{multline}
    \max\left(\frac12 \left(\frac {p(1-p)}2\right)^{N/2},\frac14\left(\frac{K_2}{M}\right)^N\right) \leq P_e^*(M) \\ \leq (N+1)4^N \binom{M}{K_1} \max\left(p^{N/2},\left(\frac{K_1}{M}\right)^N\right). \label{eq:adversarial_main}
    \end{multline}
    Moreover, the upper bound can be attained even when multisets are disallowed.
\label{thm:adversarial_model}
\end{theorem}
\begin{proof}
    \underline{Achievability bound:} We adopt an arbitrary outer codebook satisfying $|A_i \cap A_j| \le K_1$ in accordance with Definition \ref{def:K1}.  Compared to Theorem \ref{thm:main}, decoding is less straightforward because there may be decoded molecules that don't correspond to any that were sent.  Accordingly, we change the decoding rule, and consider estimating the message by choosing $i$ that maximizes the number of molecules seen in $A_i$ at the decoder (including repeated occurrences). 
    Supposing that the true message is $i$, we fix an arbitrary $j \ne i$ and consider the probability that the decoder outputs $j$ instead of $i$.  
    
    Before proceeding, we introduce two useful random variables.  Among the list of $N$ molecules sampled (\emph{before sequencing}), consider the subset containing only the $K_1$ molecules with the most occurrences, and let $N_1$ be the size of this subset.  Moreover, let $N_2$ be the total number of sequencing errors among the $N$ invocations of sequencing.

    We claim that if $N_1+2N_2 < N$, the above decoding rule is successful.  To see this, suppose that the decoder (incorrectly) outputs $j$ instead of $i$.  Then there must be at least as many decoded molecules that are elements of $A_i \setminus A_j$ compared to $A_j \setminus A_i$.  Those that are in $A_j \setminus A_i$ can only come from sequencing errors, so there are at most $N_2$ of them.  Moreover, there are $N-N_2$ molecules that do not go through any sequencing errors. Among them, at most $N_1$ of them can be in $A_i\cap A_j$ -- this is because $|A_i \cap A_j| \leq K_1$ (see Definition \ref{def:K1}), and due to the definition of $N_1$.  Therefore, there are at least $N - N_1 - N_2$ molecules in $A_i \setminus A_j$, so decoding succeeds if  $N_1+2N_2 < N$.

    The statement $N_1 \geq n_1$ is equivalent to the existence of a set of $K_1$ molecules (among those in $A_i$) such that the molecules in that set are sampled at least $n_1$ times. For a specific set of $K_1$ molecules, the number of times we sample from these $K_1$ molecules follows a ${\rm Binomial}(N,K_1/M)$ distribution, and thus the probability that we see at least $n_1$ of them is at most $\binom{N}{n} \left(\frac{K_1}{M}\right)^{n_1}$. Taking a union bound over all possible sets of $K_1$ molecules gives
    \begin{equation}
        \mathbb{P}(N_1 \geq n_1) \leq \binom{M}{K_1}\binom{N}{n_1} \left(\frac{K_1}{M}\right)^{n_1} \leq 2^N \binom{M}{K_1} \left(\frac{K_1}{M}\right)^{n_1}. \label{eq:N1bound}
    \end{equation}
    Moreover, since $N_2 \sim {\rm Binomial}(N,p)$, we have
    \begin{equation}
        \mathbb{P}(N_2 \geq n_2) \leq \binom{N}{n_2} p^{n_2} \leq 2^N p^{n_2}. \label{eq:N2bound}
    \end{equation}
    Therefore, we can  compute the probability that $N_1 + 2N_2 \geq N$ as follows:
    \begin{align}
        &\mathbb{P}(N_1 + 2N_2 \geq N)  \nonumber \\ &= \sum_{n=0}^{N} \mathbb{P}(N_1 = n, N_2 \ge \frac{N-n}{2})\\
        &\leq \sum_{n=0}^{N} \binom{M}{K_1} 4^N \left(\frac{K_1}{M}\right)^n  p^{(N-n)/2} \label{eq:pN1N2_ii} \\
        &\leq (N+1)4^N \binom{M}{K_1} \max\left(p^{N/2},\left(\frac{K_1}{M}\right)^N\right), \label{eq:pN1N2_iii}
    \end{align}
    where \eqref{eq:pN1N2_ii} follows from \eqref{eq:N1bound}--\eqref{eq:N2bound} and the fact that $N_1$ and $N_2$ are independent, and \eqref{eq:pN1N2_iii} follows since $\left(\frac{K_1}{M}\right)^n  p^{(N-n)/2}$ is maximized at either $n=0$ or $n=N$.  Combined with the fact that errors only occur when $N_1 + 2N_2 \geq N$, this completes the proof of the achievability part.
    
    \underline{Converse bound:} Let $i,j$ be any two messages. Suppose that when the true message sent is $i$ or $j$, the adversary designs the sequencing errors as follows.  When a molecule is sampled:
    \begin{itemize}
        \item[($\xi_0$)] With probability $1-p$, there is no sequencing error, so the output molecule equals the input molecule;
        \item[($\xi_i$)] With probability $\frac{p}2$, the output molecule is uniformly chosen in $A_i$ (including multiplicity, so the probabilities are weighted by frequency);
        \item[($\xi_j$)] With probability $\frac{p}2$, the output molecule is uniformly chosen in $A_j$.
    \end{itemize}
    Now consider the following ``bad'' events:
    \begin{enumerate}
        \item[(i)] The true message is $i$. Case ($\xi_j$) occurs for the first $N/2$ molecules, and case ($\xi_0$) occurs for the next $N/2$ molecules;
        \item[(ii)] The true message is $j$. Case ($\xi_0$) occurs for the first $N/2$ molecules, and case ($\xi_i$) occurs for the next $N/2$ molecules.
    \end{enumerate}
    Observe that in both of these cases, the first $N/2$ molecules seen by the decoder are drawn uniformly at random from $A_j$ (with replacement), and the remaining $N/2$ molecules are drawn uniformly at random from $A_i$ (with replacement).  That is, the joint distribution of molecules seen is identical in both cases.  This means that the decoder cannot distinguish between cases (i) and (ii).

    Conditioned on the message being in $\{i,j\}$, the two bad events above each occur with probability 
    \begin{equation}
            \frac12 \cdot \left(\frac{p}2\right)^{N/2} (1-p)^{N/2} = \frac12 \cdot \left(\frac{p(1-p)}2\right)^{N/2}.
            \label{eq:random_model_overall_error}
    \end{equation}
    Moreover, whenever one of these bad events occurs, the decoder cannot do better than random guessing between $i$ and $j$. Therefore, \eqref{eq:random_model_overall_error} is a lower bound for conditional the error probability.

    The preceding analysis holds true for any two messages $i$ and $j$.  To characterize the error probability averaged over all messages, we simply pair the $J$ messages arbitrarily to form $J/2$ pairs.  By the preceding analysis, conditioned on any one of these pairs containing the true message, the error probability is lower bounded by \eqref{eq:random_model_overall_error}.  Therefore, the same holds true of the overall average error probability.  This establishes the first term in the lower bound in \eqref{eq:adversarial_main}, and the second term follows directly from Theorem \ref{thm:main}.
\end{proof}
\begin{cor} \label{cor:adv_exponent}
Consider the scaling regime described in Section \ref{sec:scaling2}.  
If there exists $c>0$ such that $J \geq \exp(M^{2-\alpha+c})$, and it holds that $\frac{\log J}{\log M} \rightarrow \infty$ and $\frac{\log J}{M \log M} \to 0$, then under the adversarial sequencing error model with sequencing error probability $p = o(1)$, we have
    \begin{equation}
    \lim_{M \to \infty}\frac{\log \frac{1}{P_e^*(M)}}{N \log \min(\frac{1}{\sqrt p}, \frac{M \log M}{\log J})} = 1. \label{eq:adv_exponent}
\end{equation}
\end{cor}
\begin{proof}
    By the achievability part of Theorem \ref{thm:adversarial_model}, we have
    \begin{align}
        &\log \frac{1}{P_e^*(M)} \\ &\geq -\log \left((N+1)4^N \binom{M}{K_1}\right) + \log \min\left(\frac1{p^{N/2}},\left(\frac{M}{K_1}\right)^N\right) \\
        &= \log \min\left(\frac1{p^{N/2}},\left(\frac{M}{K_1}\right)^N\right) - \mathcal{O}(N) \\
        & \geq N \log \min\left(\frac{1}{\sqrt{p}},(1+o(1)) \log \frac{M\log M}{\log J}\right) - \mathcal{O}(N) \label{eq:cor_ach_iii} \\
        &= (1+o(1)) N \log \min\left(\frac{M \log M}{\log J},\frac{1}{\sqrt p}\right), \label{eq:cor_ach_iv}
    \end{align}
    where \eqref{eq:cor_ach_iii} uses $K_1 = \mathcal{O}\big(\frac{\log J}{\log M}\big)$ from \eqref{eq:K1_scaling}, and \eqref{eq:cor_ach_iv} uses $p=o(1)$ and $\log J = o(M \log M)$.
    
    Similarly, by the converse part of Theorem \ref{thm:adversarial_model}, we have
    \begin{align}
        &\log \frac{1}{P_e^*(M)} \\ &\leq  N \log \min \left( \sqrt{\frac2{{p(1-p)}}}, \frac{M}{K_2}\right) + \mathcal{O}(N)\\
        & \stackrel{\eqref{eq:K2_floor}}{\leq} N \log \min\left( \sqrt{\frac2{{p(1-p)}}},(1+o(1)) \frac{M\log M}{\log J}\right) + \mathcal{O}(N) \\
        &= (1+o(1)) N \log \min\left(\frac{M \log M}{\log J},\frac{1}{\sqrt p}\right),
    \end{align}
    noting that $\log \sqrt{\frac{2}{p(1-p)}} = \big(\log\frac{1}{\sqrt p}\big)(1+o(1))$ as $p \to 0$.
\end{proof}

    \subsection{Random model}

We now turn to the random model, again starting with non-asymptotic upper and lower bounds on the optimal error probability.
    
    \begin{theorem}
    Under the random sequencing error model, we have
\begin{equation}
    P_e^*(M) \leq  N^3 2^{M+3N}\cdot\max \left(\frac{K_1}{M}, p, M^{1-\alpha}\right)^{N-\frac{\log J}{(\alpha-1)\log M}},
    \label{eq:random_model_achievability}
\end{equation}
and
        \begin{equation}
    P_e^*(M) \ge \frac14 \max\left(p,\frac{K_2}{M}\right)^N
    \label{eq:random_model_converse}
    \end{equation}
    \end{theorem}
    \begin{proof}
    We start with the achievability bound.   
    Let $A_1, A_2, \ldots, A_J$ be codewords (without multisets) such that any two codewords intersect in at most $K_1$ molecules (\emph{cf.}, Definition \ref{def:K1}).  
    Consider the decoding rule that chooses $j$ to maximize the number of molecules seen (including multiplicity) in $A_j$. Suppose that the true message is $i$; we will generically use $j$ for any other message.

    We re-use the notation $N_1$ and $N_2$ from the proof of Theorem \ref{thm:adversarial_model}: $N_1$ denotes the total count of the $K_1$ molecules that are sampled the most times (before sequencing errors), and $N_2$ denotes the total number of sequencing errors.  Moreover, for each $j$, let $N_{3,j}$ be the number of sequencing errors that produce a molecule in $A_j$.

    Observe that the number of decoded molecules in $A_i$ is at least $N-N_2$, since every molecule not in $A_i$ must correspond to a sequencing error.  Whenever we see a molecule in $A_j$, there are two possibilities:
    \begin{itemize}
        \item There was no sequencing error and we sampled a molecule in $A_i \cap A_j$.  The number of times this occurs is upper bounded by the number of samples of molecules in $|A_i \cap A_j|$, which is further upper bounded by $N_1$ due to the fact that $A_i \cap A_j \leq K_1$.
        \item A sequencing error occurred and produced a molecule in $A_j$. There are $N_{3,j}$ such events by definition.
    \end{itemize}
    Hence, in order for a decoding failure to occur, there must exist some $j$ for which $N-N_2 \leq N_1 + N_{3,j}$, which is equivalent to $N_1 + N_2 + N_{3,j} \geq N$.

    We analyze $N_1$ and $N_2$ a similar manner to Theorem \ref{thm:adversarial_model} as follows.  
    If $N_1 \geq n_1$, then there must exist a set of $K_1$ elements in which their total frequency is larger than $n_1$.  For any $K_1$ specific elements, the number of samples from them is distributed as ${\rm Binomial}\big(N,\frac{K_1}{M}\big)$, and taking a union bound over all $\binom{M}{K_1}$ subsets gives
    \begin{equation}
        \mathbb{P}(N_1 \geq n_1) \leq \binom{M}{K_1}\binom{N}{n_1} \left(\frac{K_1}{M}\right)^{n_1} \leq 2^{M+N} \left(\frac{K_1}{M}\right)^{n_1} \label{eq:N1_bound}
    \end{equation}
    Moreover, since $N_2 \sim {\rm Binomial}(N,p)$, we have 
    \begin{equation}
        \mathbb{P}(N_2 \geq n_2) \leq \binom{N}{n_2}p^{n_2} \leq 2^N p^{n_2} \label{eq:N2_bound}
    \end{equation}
    Since $N_1$ and $N_2$ are independent, it follows that
    \begin{equation}
        \mathbb{P}(N_1 \geq n_1, N_2 \geq n_2) \leq 2^{M+2N} \left(\frac{K_1}{M}\right)^{n_1} p^{n_2} \label{eq:N1N2_bound}
    \end{equation}

    Given $N_1$ and $N_2$, the conditional distribution of $N_{3,j}$ is ${\rm Binomial}(N_2,M^{1-\alpha})$, since there are $N_2$ sequencing errors and each sequencing error has probability $\frac{M}{M^{\alpha}} = M^{1-\alpha}$ of generating a molecule in $A_j$. 
    Therefore,
    \begin{align}
        \mathbb{P}(N_{3,j} \geq n_3|N_1=n_1, N_2=n_2) &\leq M^{(1-\alpha)n_3} \binom{n_2}{n_3} \\  &\leq 2^N M^{(1-\alpha)n_3}.
    \end{align}
    Letting $N_3 = \max_j N_{3,j}$, the union bound over the $J-1$ choices of $j$ gives
    \begin{equation}
        \mathbb{P}(N_3 \geq n_3|N_1=n_1, N_2=n_2) \leq \min(1, J \cdot 2^N M^{(1-\alpha)n_3}). \label{eq:N3_bound}
    \end{equation}
    It follows that
    \begin{align}
        &\mathbb{P}(N_1 + N_2 + N_3 \geq N)\\ &\leq \sum_{n_1+n_2+n_3\geq N } \mathbb{P}(N_3 \geq n_3 | N_1 = n_1, N_2 = n_2) \nonumber \\
            &\hspace*{4.5cm} \times \mathbb{P}(N_1 =n_1, N_2 = n_2)\\
        &\leq \sum_{n_1+n_2+n_3\geq N} \min(1, J \cdot 2^N M^{(1-\alpha)n_3}) 2^{M+2N} \nonumber \\
            &\hspace*{4.5cm} \times \left(\frac{K_1}{M}\right)^{n_1} p^{n_2}, \label{eq:PN1N2N3}
    \end{align}
    where the last step combines \eqref{eq:N1N2_bound} and \eqref{eq:N3_bound}.

    We define a threshold $\gamma = \frac{\log J}{(\alpha-1)\log M}$, and split the summation in \eqref{eq:PN1N2N3} into two cases:
    \begin{itemize}
        \item \emph{Case 1} ($n_3 \leq \gamma$). In this case, the condition $n_1 + n_2 + n_3 \geq N$ implies that $n_1 + n_2 \geq N - \gamma$ and we deduce that
        \begin{multline}
            \min(1, J \cdot 2^N M^{(1-\alpha)n_3}) 2^{M+2N} \left(\frac{K_1}{M}\right)^{n_1} p^{n_2}  \\ \leq 2^{M+2N} \max\left(\frac{K_1}{M}, p\right)^{N-\gamma},
        \end{multline}
        where we used $\big(\frac{K_1}{M}\big)^{n_1}p^{n_2} \le \big(\max\big\{\frac{K_1}{M},p\big\}\big)^{n_1 + n_2}$ followed by $n_1 + n_2 \ge N-\gamma$.
        \item \emph{Case 2} ($n_3 > \gamma$). Re-arranging the definition of $\gamma$ gives
        \begin{equation}
            J = M^{(\alpha-1)\gamma}, \label{eq:J_rearranged}
        \end{equation}
        which implies the following:
        \begin{align}
            &\min(1, J \cdot 2^N M^{(1-\alpha)n_3}) 2^{M+2N} \left(\frac{K_1}{M}\right)^{n_1} p^{n_2} \nonumber \\
            &\quad \stackrel{\eqref{eq:J_rearranged}}{\leq} 2^{M+3N} M^{(1-\alpha)(n_3-\gamma)} \left(\frac{K_1}{M}\right)^{n_1} p^{n_2}\\
            &\quad \leq 2^{M+3N} \cdot \max\left(\frac{K_1}{M}, p, M^{(1-\alpha)}\right)^{N-\gamma}, 
        \end{align}
        where the last inequality comes from the fact that $n_3-\gamma$, $n_1$, and $n_2$ are all non-negative and the condition $n_1+n_2+n_3 \geq N$ implies $(n_3-\gamma) + n_1 + n_2 \geq N-\gamma$.  
    \end{itemize}  

    Combining the two cases with \eqref{eq:PN1N2N3}, we get 
    \begin{multline}
        \mathbb{P}(\exists j \text{ s.t. } N_1 + N_2 + N_{3,j} \geq N) \\ 
        \leq  N^3 \cdot 2^{M+3N} \cdot \max\left(\frac{K_1}{M}, p, M^{(1-\alpha)}\right)^{N-\gamma}
    \end{multline}
    and substituting $\gamma = \frac{\log J}{(\alpha-1)\log M}$ gives the desired result. 
    
    Regarding the converse part, this bound of $\frac14 \big( \max\big\{p,\frac{K_2}{M}\big\}\big)^N$ was already established for the erasure model, and it immediately also applies here due to the fact that the decoder in the erasure model could choose to replace each erasure by a random molecule. 
\end{proof}

\begin{cor} \label{cor:rand_exponent}
Consider the scaling regime described in Section \ref{sec:scaling2}.  
Suppose that there exists $c>0$ such that $J \geq \exp(M^{2-\alpha+c})$, and it holds that $\frac{\log J}{\log M} \rightarrow \infty$ and $\frac{\log J}{M \log M} \to 0$. Then, under the random sequencing error model with sequencing error probability $p = o(1)$, we have
\begin{equation}
    \lim_{M \to \infty} \frac{\log \frac{1}{P_e^*(M)}}{N \log \min(\frac{1}{p},\frac{M \log M}{\log J})} = 1. \label{eq:rand_corollary}
\end{equation}
\end{cor}
\begin{proof}
For the achievability bound, using \eqref{eq:random_model_achievability}, we have
    \begin{align}
        &\frac{1}{N} \log \frac{1}{P_e^*(M)} \geq \frac{N -\frac{\log J}{(\alpha-1)\log M}}{N}\log \min \left(\frac{M}{K_1}, \frac1p, M^{\alpha-1}\right) \nonumber \\ &\hspace*{4cm} - \frac1N\log (N^3 2^{M+3N}).
        \label{eq:log_random_model_ach}
    \end{align}
    To handle the term $\frac{N-\frac{\log J}{(\alpha-1)\log M}}{N}$, observe that
    \begin{equation}
        \frac{\log J}{(\alpha-1)\log M} = \frac{o(M \log M)}{(\alpha-1)\log M} = o(M) = o(N),
    \end{equation}
    so that
    \begin{equation}
        \frac{N -\frac{\log J}{(\alpha-1)\log M}}{N} = 1+o(1).
        \label{eq:n-n'}
    \end{equation}
Regarding $K_1$, we observe from \eqref{eq:K1_scaling} that
\begin{equation}
        K_1 = \mathcal{O}\Big(\frac{\log J}{\log M}\Big) \Rightarrow \log \frac{M}{K_1} \geq \log \frac{M \log M}{\log J} - \mathcal{O}(1).
        \label{eq:k1_bound_error}
    \end{equation}
    Finally, regarding the final term in \eqref{eq:log_random_model_ach}, we have
\begin{equation}
    \frac1N\log (N^3 2^{M+3N}) = \mathcal{O}(1).
    \label{eq:remaining_error_term}
\end{equation}
Substituting \eqref{eq:n-n'}, \eqref{eq:k1_bound_error} and \eqref{eq:remaining_error_term}  into \eqref{eq:log_random_model_ach}, we obtain
    \begin{equation}
        \frac{1}{N} \log \frac{1}{P_e^*(M)}
        \geq (1+o(1))\log \min \left(\frac{M\log M}{\log J}, \frac1p, M^{\alpha-1}\right).\label{eq:term_to_drop}
    \end{equation}
To simply this expression, we recall the assumption $J \geq \exp(M^{2-\alpha+c})$, and observe that the following holds for sufficiently large $M$:
\begin{equation}
    \frac{M\log M}{\log J} \leq \frac{M\log M}{M^{2-\alpha+c}} = M^{\alpha-1-c}\log M <M^{\alpha-1},
\end{equation}
which means that the $M^{\alpha-1}$ term in \eqref{eq:term_to_drop} can be dropped. This completes the proof of the achievability part.

    The converse part follows from an identical argument to that of Corollary \ref{cor:erasure_exponent} (or alternatively, the converse part of Corollary \ref{cor:erasure_exponent} for the erasure model directly implies the same for the random model).
\end{proof}

\subsection{Comparison of Models} \label{sec:compare_models}

It is evident from the definitions of the noise models that the ordering of exponents from smallest to largest should be as follows: adversarial, random, then erasures.  Our results in Corollaries \ref{cor:erasure_exponent}, \ref{cor:adv_exponent}, and \ref{cor:rand_exponent} are consistent with this, but perhaps surprisingly, the exponents for the random and erasure models turn out to be the same.  To interpret this in more detail, we can consider three types of error that we saw throughout the proofs:
\begin{itemize}
    \item We may only sample molecules in the intersection of two codewords $A_i$ and $A_j$;
    \item A sequencing error may occur in every molecule, in which case the output reveals no information about the message.  
    \item When the message is $A_i$, half the molecules in $A_i$ may undergo a sequencing error with each of them producing a molecule in $A_j$ (for some $j \ne i$), and this is indistinguishable from the an analogous scenario with the roles of $A_i$ and $A_j$ reversed.
\end{itemize}
The first of these types of error is present even without sequencing errors, and thus appears for all 3 noise models (see \eqref{eq:erasure_exponent}, \eqref{eq:adv_exponent}, and \eqref{eq:rand_corollary}).  The second type of error is also present under the erasure and random models (and the adversarial model, but it is never dominant there).   
The third type of error may occur under both the random and adversarial models, but with a major difference: In the random noise model, consistently producing molecules from $A_j$ needs to happen by chance, but in the adversarial model, the adversary can simply make that happen directly.  Accordingly, we get $\frac{1}{\sqrt p}$ in \eqref{eq:adv_exponent}, whereas for the random model, the analogous term would be $\frac{1}{\sqrt p} \cdot M^{(\alpha-1)/2}$. Similar to the argument following \eqref{eq:term_to_drop}, we can show that such a term is never dominant, at least when $J \geq \exp(M^{2-\alpha+c})$.  It is conceivable that more substantial differences between the models would arise when $J \ll \exp(M^{2-\alpha+c})$, but we leave such considerations for possible future work.

We can also identify regimes in which all three noise models give the same exponent.   Recall that these results assume that $J \ge \exp(M^{2 - \alpha + c})$, $J \le e^{o(M \log M)}$, and $p = o(1)$.  If we fix a decay rate for $p$ (e.g., $p = M^{-0.01}$) and consider various scaling laws for $J$ between its upper and lower limits, we see that whenever $J$ is ``sufficiently close'' to its upper limit (namely, we have $J = e^{o(M \log M)}$ but with $o(\cdot)$ decaying slowly enough), it holds for all three sequencing error models that 
\begin{equation}
    \lim_{M \to \infty} \frac{\log \frac{1}{P_e^*(M)}}{N \log \frac{M \log M}{\log J}} = 1.
\end{equation}
Intuitively, $J$ being close to its upper limit means being ``closer to a non-zero rate'', so the fact that all three models give the same exponent is consistent with our main result Theorem \ref{thm:main_result} for the constant-rate regime (in which the noise model plays no role).





\section{Conclusion}
 
We have derived exact error exponents for a concatenated coding based class of DNA storage codes, and showed significant improvements over an existing achievable exponent.  We found that the regime of a constant rate and a super-linear number of reads permits a particularly simple error exponent, whereas the low-rate regime comes with a number of additional intricacies such as the suboptimality of having distinct molecules and the emergence of dependence on the sequencing error model.  
Possible directions for future research include devising more efficient decoding schemes (e.g., maximizing $|S \cap A_i|$ in Section \ref{sec:ach} is likely to be intractable) and further studying the error exponents of other classes of DNA storage codes, particularly ones that can attain higher achievable rates than concatenated codes (as is known to be information-theoretically possible \cite{merhav}).

\appendices

\input{appendix}

\bibliographystyle{IEEEtran}
\bibliography{general}

     \begin{IEEEbiographynophoto}{Yan Hao Ling}
        received the B.Comp.~degree in computer science and the B.Sci.~degree 
        in mathematics in 2021, and the PhD degree in computer science in 2025, all from the National University of Singapore (NUS).  
        His research interests are in the areas of
        information theory, statistical learning, and theoretical computer science.
    \end{IEEEbiographynophoto}
    
     \begin{IEEEbiographynophoto}{Jonathan Scarlett}
        (S'14 -- M'15) received 
        the B.Eng.~degree in electrical engineering and the B.Sci.~degree in 
        computer science from the University of Melbourne, Australia. 
        From October 2011 to August 2014, he
        was a Ph.D. student in the Signal Processing and Communications Group
        at the University of Cambridge, United Kingdom. From September 2014 to
        September 2017, he was post-doctoral researcher with the Laboratory for
        Information and Inference Systems at the \'Ecole Polytechnique F\'ed\'erale
        de Lausanne (EPFL), Switzerland. Since January 2018, he has been with the Department of Computer Science and Department of Mathematics at the
        National University of Singapore, where he is currently an Associate Professor. His research interests are in
        the areas of information theory, machine learning, signal processing, and
        high-dimensional statistics. He received the Singapore National Research Foundation (NRF)
        fellowship, and the NUS Presidential Young Professorship award.
    \end{IEEEbiographynophoto}

\end{document}

%% file: isit_part.tex
\section{Introduction}

In recent years, significant research attention has been paid to characterizing the capacity of DNA storage systems; see \cite{shomorony2022information} for a recent overview. In contrast, only limited attention has been paid to error exponents, which seek a more precise characterization of the error probability by considering its exponential decay at rates below capacity, and have long been studied in standard channel coding problems \cite{gallager1968information,csiszar2011information}.

Two recent studies concerning the error exponents of DNA storage codes are \cite{merhav} and \cite{Weinberger}. The work of Merhav and Weinberger \cite{merhav} adopts a coding strategy that relies on random coding, which is a powerful theoretical tool but is highly impractical.  Moreover, their study is specific to discrete memoryless sequencing channels, meaning that there are only substitution errors and no insertions or deletions.
Motivated by these limitations, Weinberger \cite{Weinberger} studied achievable error exponents for a class of concatenated codes in which an ``inner code'' for the sequencing channel is used in a black-box manner for individual molecules, and an ``outer code'' is used to handle the entire set of molecules.

In this paper, we consider the same class of codes as \cite{Weinberger}, but provide \emph{exact} error exponents via matching achievability and converse bounds. Our achievability results strictly improve on \cite{Weinberger} despite having a somewhat simpler analysis; a notable weakness in the analysis of \cite{Weinberger} is using a Poisson approximation to the multinomial distribution (see also Section \ref{sec:poisson}). Our converse results appear to be new, though they are related to a discussion item in \cite[p.~7010, item 6]{Weinberger}. 

While we obtain exact exponents for the class considered, we note that this class itself can be suboptimal, as it precludes certain advanced techniques such as clustering \cite[Sec.~5.1]{shomorony2022information} and ``full'' random coding \cite{merhav}. In particular, the exponent is only positive for rates up to a certain threshold that can be strictly worse than the one in \cite{merhav} (see \eqref{eq:rate_achieved} below).

Similarly to the related works \cite{merhav,Weinberger}, our work is complementary to the extensive work on \emph{coding-theoretic} considerations for DNA storage codes (e.g., see \cite{lenz2019coding,song2020sequence,kovavcevic2018codes} and the references therein), which typically seek distinct goals such as good distance properties.

\section{Model and Definitions}

\subsection{The DNA Storage Model}

We follow the same setup as \cite{Weinberger}.  
The encoder is first given a message $m \in \{1,2,\ldots, \exp(RML)\}$,\footnote{Throughout the paper, we ignore rounding issues for quantities such as $\exp(RML)$, as this does not impact the results.} where $R > 0$ represents the coding rate.  Given the message, the encoder outputs a multiset $A_m$ of $M$ molecules, each of length $L$ with symbols coming from some alphabet $\mathcal{X}$ (e.g., $\mathcal{X} = \{A,C,G,T\}$).  The output received by the decoder is then generated as follows:
\begin{itemize}
    \item \emph{Sampling:} $N$ molecules are sampled uniformly at random with replacement from $A_m$.
    \item \emph{Sequencing:} For each molecule $x^L$ sampled (or $x$ for short), the decoder receives an output $y^{(L)}$ (or $y$ for short) generated randomly according to some sequencing channel.  It is assumed that the $N$ uses of the sequencing channel are independent with the same transition law.
\end{itemize}
Although the $N$ uses of the sequencing channel are independent, this channel itself may follow an arbitrary conditional distribution with inputs in $\mathcal{X}^L$ and outputs in some alphabet $\mathcal{Y}^{(L)}$.  In particular, $\mathcal{Y}^{(L)}$ is not necessarily a Cartesian product, and this allows us to cater for different kinds of sequencing channels, such as ones with insertions and deletions.

The decoder is given the $N$ outputs $(y_1, y_2,\ldots, y_N)$, and forms an estimate $\hat{m}$ of the original message.  The average error probability is denoted by $P_e = \mathbb{P}(\hat{m} \ne m)$ with $m$ being uniformly random over $\{1,2,\ldots, \exp(RML)\}$.

\subsection{Concatenated Coding Based Class of Protocols}

We now describe the concatenated coding based class of DNA storage codes that we consider.  This class is motivated by previous practical and theoretical uses of concatenated codes for DNA storage (e.g., \cite{lenz2020achievable,ren2022dna}), and we refer the reader to \cite{Weinberger} for further discussion on the practical motivation.

An \emph{inner code} $(X,D)$ with parameters $(\rin, L)$ is given by the following:
\begin{itemize}
    \item An inner codebook of $\exp(\rin L)$ molecules, each of length $L$, which we denote as $X = (x_1, x_2, \dots, x_{\exp(\rin L)})$.  We make the mild assumption that these molecules are all distinct.
    \item An inner decoding function operating on the sequencing channel output, $D: \mathcal{Y}^{(L)} \rightarrow \{x_1,\ldots, x_{\exp(\rin L)}\}$.
\end{itemize}
Given the inner code $(X,D)$, for each $x \in X$, the (inner) error probability for $x$ is the probability that $D(y) \neq x$, where $y$ is distributed according to the sequencing channel with input $x$.  The highest (among all $x \in X$) of these error probabilities is called the \emph{maximal error probability} of the inner code.\footnote{For the inner code, the maximal error turns out to be more convenient than average error.  Mathematically, the latter readily leads to the former via a standard expurgation argument \cite[Sec.~7.7]{cover2006elements}.}

We will use the inner code in a ``black-box'' manner, only assuming (except where stated otherwise) that it has maximal error probability approaching zero as $L$ increases.  This motivates the following definition:

\begin{dfn} \label{def:concat_class}
A sequence of inner codebooks $(X_L, D_L)_{L=1}^\infty$ \emph{achieves a rate $\rin$} if each $(X_L,D_L)$ has parameters $(\rin, L)$, and the maximal error probability of $(X_L, D_L)$ approaches zero as $L \to \infty$. Such a rate is said to be \emph{achievable}.
\end{dfn}

Next, we formally state the class of concatenated codes that we consider in this paper.

\begin{dfn}
A protocol with parameters $(M,L,N,R)$ and inner code $(X,D)$ is said to perform \emph{separate inner and outer coding} if it satisfies the following two properties:
\begin{itemize}
    \item For any message $m$ at the encoder, the resulting input is a multiset $A_m$ of size $M$ whose elements are chosen from the inner codebook $X$.
    \item After the decoder samples and sequences $N$ molecules to obtain $y_1, y_2,\ldots, y_N$, the estimate of the message depends only on $D(y_1), D(y_2),\ldots, D(y_N)$.
\end{itemize}
\label{dfn:use_inner}
\end{dfn}

Observe that under this class of codes, we can view sequencing and inner decoding as a single step: For any two molecules $x, x'$, the sequencing channel and $D$ together determine a transition probability $P(x'|x)$, where $x'$ represents the decoded codeword.  
We can then summarize the entire concatenated coding based protocol by the following steps:
\begin{itemize}
    \item \emph{Encoding:} For each possible message $m$, there is an outer codeword $A_m$, which is a multiset containing $M$ molecules from the inner codebook $X$.
    \item \emph{Sampling:} The decoder samples $N$ molecules following the multinomial distribution over the multiset $A_m$ (with probability $\frac{1}{M}$ for each element).
    \item \emph{Sequencing and inner decoding:} For each input molecule $x_i$, the decoder receives an output molecule $x'_i$ following the transition probability $P$.
    \item \emph{Decoding:} The decoder forms an estimate $\hat{m}$, which depends only on $(x'_1, x'_2, \ldots, x'_N)$ (and the outer codebook).
\end{itemize}
If we are using a sequence of codebooks that achieves the rate $\rin$, then the transition law $P(x'|x)$ satisfies $P(x | x) \to 1$ as $L \to \infty$ for all $x \in X$; we will require this in our achievability part, but our converse will be more general. 

In \cite{Weinberger}, the code was further assumed to be \emph{index-based} according to the following definition.

\begin{dfn}
    A codebook is \emph{index-based} if there exist $M$ disjoint sets of molecules $(B_i)_{i=1}^M$ of equal size such that every outer codeword $A_m$ contains exactly one molecule from each $B_i$.
    \label{dfn:index_based}
\end{dfn}

Index-based codes are often considered to be favorable for keeping the encoder and decoder simple.  With the exception of some of our results for the low rate regime, our achievability results will use index-based coding and will match our converse results that have no such requirement, thus showing that index-based codes attain optimal error exponents within the broader class of codes given in Definition \ref{dfn:use_inner}.

\subsection{Scaling of Parameters} \label{sec:scaling}

In principle, there any many possibilities for how $(L,M,N)$ scale with respect to one another.  We will focus our attention on scaling regimes that are the most widely-adopted in information-theoretic studies (e.g., \cite{Weinberger,merhav,fundamental_limit_dna_2017}), and are practically well-motivated (e.g., $L = \Theta(\log M)$ corresponding to relatively short reads).

We consider the limit $M \rightarrow \infty$, and in the first part of the paper, we assume that the other quantities scale in manner that keeps the following parameters constant:
\begin{itemize}
    \item The coverage depth, which we denote as $c = N/M$;
    \item The molecule length parameter $\beta$, for which the length of the molecules grows as $L = \beta \log M$;
    \item The inner rate $\rin$, such that the inner codebook $X$ has size $\exp(\rin L)$;
    \item The outer rate $R$ corresponding to having $\exp(RML)$ messages.
\end{itemize}
For any given value of $M$, the number of molecules in a codebook with parameters $(\rin ,L) = (\rin, \beta \log M)$ is $\exp(\rin L) = M^{\beta \rin}$. 
For notational convenience, we let 
\begin{equation}
\alpha = {\beta \rin}
    \label{eq:alpha_def}
\end{equation} 
so that the number of molecules in the inner codebook is $M^\alpha$.  Note that for index-based codes, this implies each $B_i$ in Definition \ref{dfn:index_based} having size $M^{\alpha - 1}$ (for $\alpha > 1$).

Under the preceding scaling laws, the \emph{capacity} is defined as the supremum of $R$ for which there exists a sequence of codes (indexed by $M$) attaining $P_e \to 0$. 
Under a simpler model in which the $M$ sampled molecules are observed directly  in a uniformly random order with no sequencing errors, the capacity is as follows when $\alpha > 1$ \cite[Lemma 1]{fundamental_limit_dna_2017}:\footnote{The result in \cite{fundamental_limit_dna_2017} considers a binary code where $\rin = 1$, but the proof can easily be adapted to obtain this more general version.}
\begin{equation}
    \rin - \frac{1}{\beta} = \frac{\alpha-1}{\beta}. \label{eq:capacity}
\end{equation}
When $\alpha < 1$, the capacity is zero, and moreover, indexed-based coding according to Definition \ref{dfn:index_based} becomes impossible because there are fewer than $M$ molecules to begin with.  Accordingly, throughout the paper we will only consider the case that $\alpha > 1$.

It turns out to be more convenient to express the outer rate of the protocol as a fraction of \eqref{eq:capacity}:
\begin{equation}
    \rbar = \frac{R}{\rin - \frac{1}{\beta}} = \frac{R \beta}{\alpha - 1}.
    \label{eq:rbar_def}
\end{equation}
Thus, the number of possible messages that the encoder can receive is 
\begin{equation}
    \exp(RML) = \exp((\alpha-1)\rbar M\log M). \label{eq:RML}
\end{equation}
In the first part of the paper, we will treat all of $(c,\rin,R,\rbar,\alpha,\beta)$ as constants. 

Afterwards, in Section \ref{sec:superlinear}, we will consider the case that $N = \omega(M)$, i.e., a super-linear number of reads, which is motivated by the fact that reads are often inexpensive.  Then, in Sections \ref{sec:lowrate_no_seq} and \ref{sec:lowrate_seq}, we will consider scenarios where the number of messages scales as $\exp(o(M \log M))$, which can roughly be viewed as ``zooming in'' to the low-rate regime ($R \to 0$ above).  This turns out to be a significantly more intricate regime, with different error events being dominant and the model for sequencing errors playing a crucial role.

In the remainder of the paper, we let $P_e^*(M)$ denote the optimal error probability among all 
protocols performing separate inner and outer coding (see Definition \ref{dfn:use_inner}), where the sequence of inner codebooks $(X_L, D_L)$ with rate $\rin$ is also optimized.  Our goal is to establish the optimal error exponent $\lim_{M \to \infty} \frac{1}{M}\log\frac{1}{P_e^*(M)}$ (when $N = cM$) or $\lim_{M \to \infty} \frac{1}{N}\log\frac{1}{P_e^*(M)}$ (which turns out to be the appropriate normalization when $N = \omega(M)$).\footnote{When $N=cM$ for fixed $c > 0$, it is inconsequential whether we normalize by $N$ or $M$, as we end up with the same exponent up to multiplication or division by $c$.  We choose to normalize by $M$ for consistency with \cite{merhav,Weinberger}.}

Since our results are spread out throughout the entire paper, an overview is provided in Table \ref{tbl:summary} for convenience.

\begin{table*}
    \centering
    \caption{Overview of our results.  Sequencing errors occur independently with probability $p=o(1)$, in which case the decoder sees an arbitrary molecule (`Adversarial'), a uniformly random molecule (`Random') from the inner code, or no molecule (`Erasure').  An entry of `Any' means the result applies to all of these models, and an entry of `None' means no sequencing errors.  The constant $\alpha > 1$ is the value such that the inner code has size $M^{\alpha}$.  The final 6 rows assume $J = e^{o(M \log M)}$ (low-rate) and $\frac{\log J}{\log M} \to \infty$ (not too few messages).
    \label{tbl:summary}}

    \begin{tabular}{|c|c|c|c|c|}
    \hline 
    Result & Number of Messages $J$ & Scaling of $\log\frac{1}{P_e^*(M)}$ & Sequencing Error Model & Notes\tabularnewline
    \hline 
    \hline 
    Theorem \ref{thm:main_result} & $e^{\Theta(M\log M)}$ & $\Theta(M)$ & Any & $N=\Theta(M)$\tabularnewline
    \hline 
    Theorem \ref{thm:superlinear}  & $e^{\Theta(M\log M)}$ & $\Theta(N)$ & Any & $N=\omega(M)$\tabularnewline
    \hline 
    Theorem \ref{thm:main_exponent} & $\gtrsim\exp(M^{2-\alpha})$ & $\Theta\big(N\log\frac{M\log M}{\log J}\big)$ & None & \tabularnewline
    \hline 
    Corollary \ref{cor:no_multi} & $\lesssim\exp(M^{2-\alpha})$ & $\Theta( N \log M )$ & None & $\alpha\in(1,2)$, no multi-sets\tabularnewline
    \hline 
    Theorem \ref{thm:multi_allowed} & $\lesssim\exp(M^{2-\alpha})$ & $\Theta( N \log M )$ & None & $\alpha\in(1,2)$, multi-sets allowed\tabularnewline
    \hline 
    Corollary \ref{cor:erasure_exponent} & $\gtrsim\exp(M^{2-\alpha})$ & $\Theta\big( N \log \min( \frac1p, \frac{M\log M}{\log J}) \big)$ & Erasure & \tabularnewline
    \hline 
    Corollary \ref{cor:adv_exponent} & $\gtrsim\exp(M^{2-\alpha})$ & $\Theta\big( N \log \min(\frac{1}{p}, \frac{M \log M}{\log J}) \big)$ & Adversarial & \tabularnewline
    \hline 
    Corollary \ref{cor:rand_exponent} & $\gtrsim\exp(M^{2-\alpha})$ & $\Theta\big( N \log \min(\frac{1}{p},\frac{M \log M}{\log J}) \big)$ & Random & \tabularnewline
    \hline 
    \end{tabular}
\end{table*}

\subsection{Statement of First Main Result}

Our first main result is written in terms of a key combinatorial quantity $p(N,M,K)$, which we define as follows: If we take $N$ independent samples uniformly at random from the set $\{1,2,\dotsc,M\}$, then
\begin{equation}
    p(N,M,K) = \mathbb{P}\big( \text{\#distinct samples} \le K\big). \label{eq:def_p}
\end{equation}
For example, if $M = 6$ and $N = 8$, and the samples are $(2,6,2,1,6,4,4,6)$, then there are 4 distinct samples, namely, $\{1,2,4,6\}$.  Observe that this sampling procedure coincides with that done in the sampling step in our problem setup (before accounting for sequencing errors).  The quantity $p(N,M,K)$ will be characterized in Section \ref{sec:balls_bins}.

\begin{theorem} \label{thm:main_result}
Consider the scaling regime described in Section \ref{sec:scaling}.  
Fix $c>0$ and $\rbar \in (0,1)$, and let $\rin$ be any achievable rate for the sequencing channel.\footnote{For the converse part, even if $\rin$ exceeds the capacity of the sequencing channel, \eqref{eq:main_result} is still an upper bound on the error exponent.  However, in view of \eqref{eq:rbar_def}, the result is weaker compared to the case of achievable $\rin$.}  Then $P_e^*(M)$ has the same exponential dependence as $p(cM, M, \rbar M)$:
\begin{equation}
    \lim_{M \rightarrow \infty} -\frac1M \log P_e^*(M) = \lim_{M\rightarrow \infty} -\frac1M \log p(cM, M, \rbar M) \label{eq:main_result}
\end{equation}
with $R_0$ defined in \eqref{eq:rbar_def}.  
Furthermore, there exist index-based codebooks (Definition \ref{dfn:index_based}) that achieve this exponent.
\end{theorem}
\begin{proof}
    See Section \ref{sec:const_rate}.
\end{proof}

In Section \ref{sec:poisson}, including Figure \ref{fig:poisson_vs_multinomial} therein, we will compare this to the achievability result of Weinberger \cite{Weinberger}, showing a significant improvement (particularly at low rates) and discussing a related Poisson sampling model.  We also note that $p(cM,M,R_0M)$ is increasing in $R_0$ by definition, so if one has the flexibility to choose the rate of the inner code, it should be chosen as close as possible to the capacity of the sequencing channel in order to decrease $R_0$ (see \eqref{eq:rbar_def}).

Next, we proceed to make the right-hand side of \eqref{eq:main_result} more explicit.

\subsection{Characterizing $p(N,M,K)$ via the Balls and Bins Problem} \label{sec:balls_bins}

The quantity $p(N,M,K)$ can be viewed as coming from a balls and bins problem, where there are $N$ balls, we independently throw each of them into one of $M$ bins chosen uniformly at random, and $p(N,M,K)$ is the probability of having at most $K$ non-empty bins.

In the following, we demonstrate the existence of the limit on the right-hand side of \eqref{eq:main_result}, 
and give a formula for it.

\begin{theorem}
    For all $c>0$ and $0< \delta < 1$, the limit
    \begin{equation}
    \lim_{M \rightarrow \infty} -\frac1M \log p(cM, M, \delta M)
\end{equation}
exists as a function $f(c,\delta)$.  Furthermore, $f(c,\delta)$ is continuous in $\delta$ for any fixed $c$, and is given as follows:
\begin{itemize}
    \item[(i)] If $1-\exp(-c) \geq \delta$, then there exists a unique $r \in (\delta, 1]$ with
    \begin{equation}
        1-\exp\Big(-\frac{c}{r}\Big) = \frac{\delta}{r}, \label{eq:r}
    \end{equation}
    and it holds that
    \begin{equation}
    f(c,\delta)=
             -c\log r - H_2(\delta) + r H_2\Big(\frac{\delta}{r}\Big),
    \label{eq:f_eval}
    \end{equation}
    where $H_2(x) := x \log (1/x) + (1-x) \log (1/(1-x))$ is the binary entropy function.
    \item[(ii)] If $1-\exp(-c) < \delta$, then $f(c,\delta)=0$.
\end{itemize}
\label{thm:f_formula}
\end{theorem}

While many aspects of the balls and bins problem are well-studied in the literature, we were unable to find any existing work giving this result.  
We thus provide the proof in Appendix \ref{app:balls_bins_pf} using a direct combinatorial argument along with some asymptotic analysis.  We briefly mention that the logic behind the choice of $r$ in \eqref{eq:r} is that it maximizes the right-hand side of \eqref{eq:f_eval}.

It is interesting to observe what happens in two limiting regimes (after having already taken $M \to \infty$):
\begin{itemize}
    \item Suppose that $c \to \infty$ for fixed $\delta$.  Then by \eqref{eq:r} we get $r \to \delta$, and by \eqref{eq:f_eval} we get 
    \begin{equation}
        f(c,\delta) = c\log\frac{1}{\delta} + \mathcal{O}(1). \label{eq:lim_c_large}
    \end{equation}
    Thus, the error exponent is dominated by $c\log\frac{1}{\delta}$.
    \item Suppose that $\delta \to 0$ for fixed $c$.  Then by \eqref{eq:r} we get $r \to 0$ and $\frac{\delta}{r} \to 1$.  Substituting into \eqref{eq:f_eval} then gives
    \begin{equation}
        f(c,\delta) = c\log\frac{1}{\delta} + o(1). \label{eq:lim_delta_small}
    \end{equation}
    Thus, the error exponent is again dominated by $c\log\frac{1}{\delta}$.
\end{itemize}
It is also interesting to consider which combinations of $c$ and $\delta$ give $f(c,\delta) > 0$, i.e., a positive error exponent.  It is straightforward to check that the transition between $f(c,\delta)$ being zero and positive occurs when $\delta = 1-\exp(-c)$.  Hence, and by choosing $R_{\rm in}$ arbitrarily close to the sequencing error channel capacity $C_{\rm in}$ in \eqref{eq:rbar_def}, we can attain a positive error exponent in Theorem \ref{thm:main_result} whenever
    \begin{equation}
        R < (1-e^{-c})\Big( C_{\rm in} - \frac{1}{\beta} \Big). \label{eq:rate_achieved}
    \end{equation}
    The fact that any such rate is achievable via index-based concatenated codes is well-known (e.g., see \cite[Sec.~5]{shomorony2022information}), and we re-iterate that the right-hand side can be strictly smaller than the capacity under arbitrary codes.

\subsection{Comparison with the Poisson Sampling Model and Existing Work} \label{sec:poisson}

Closely related to the multinomial sampling model is the Poisson sampling model, in which the number of times each molecule is sampled is independently drawn from a ${\rm Poisson}\big(\frac{N}{M}\big)$ distribution.  Hence, the total number of sampled molecules is ${\rm Poisson}(N)$, instead of being fixed to $N$ as in the multinomial distribution.  While there are well-known results showing that multinomial and Poisson distributions are ``close'' (e.g., \cite{arenbaev1977asymptotic}), their large deviations behavior can be substantially different, leading to different error exponents.

Recall the balls-and-bins interpretation from Section \ref{sec:balls_bins}.  
Under the Poisson sampling model, the probability that a specific bin is empty is simply $\exp(-c)$, independent of all other bins. Therefore, the number of non-empty bins follows a binomial distribution, and a standard Chernoff-style argument gives that $f(c,\delta)$ from Theorem \ref{thm:f_formula} is replaced by the KL divergence $f(c,\delta) = D(1-\delta \| \exp(-c))$ whenever $\delta \leq 1-\exp(-c)$. 
Our achievability and converse proofs are not affected when we change to the Poisson sampling model, as they are expressed in terms of $f$ and do not depend on the specific function $f$ -- only monotonicity and continuity are required.

At low rates, the Poisson model has significantly worse error exponents; the difference between these two models can be seen in Figure \ref{fig:poisson_vs_multinomial}.  In particular, we know from \eqref{eq:lim_delta_small} that the multinomial error exponent grows unbounded as $\delta \to 0$ (i.e., $R_0 \to 0$, since we set $\delta=R_0$), but this is not the case under the Poisson model, where as $\delta \rightarrow 0$, we simply get $D(1-\delta \| \exp(-c)) \rightarrow c$. 
%
We note that the achievable error exponent derived in \cite{Weinberger} is precisely that of the Poisson sampling model, which gives a fairly loose bound under the multinomial model (especially at low rates) in view of the above discussion.

\begin{figure}
    \centering 
    \includegraphics[width=0.95\columnwidth]{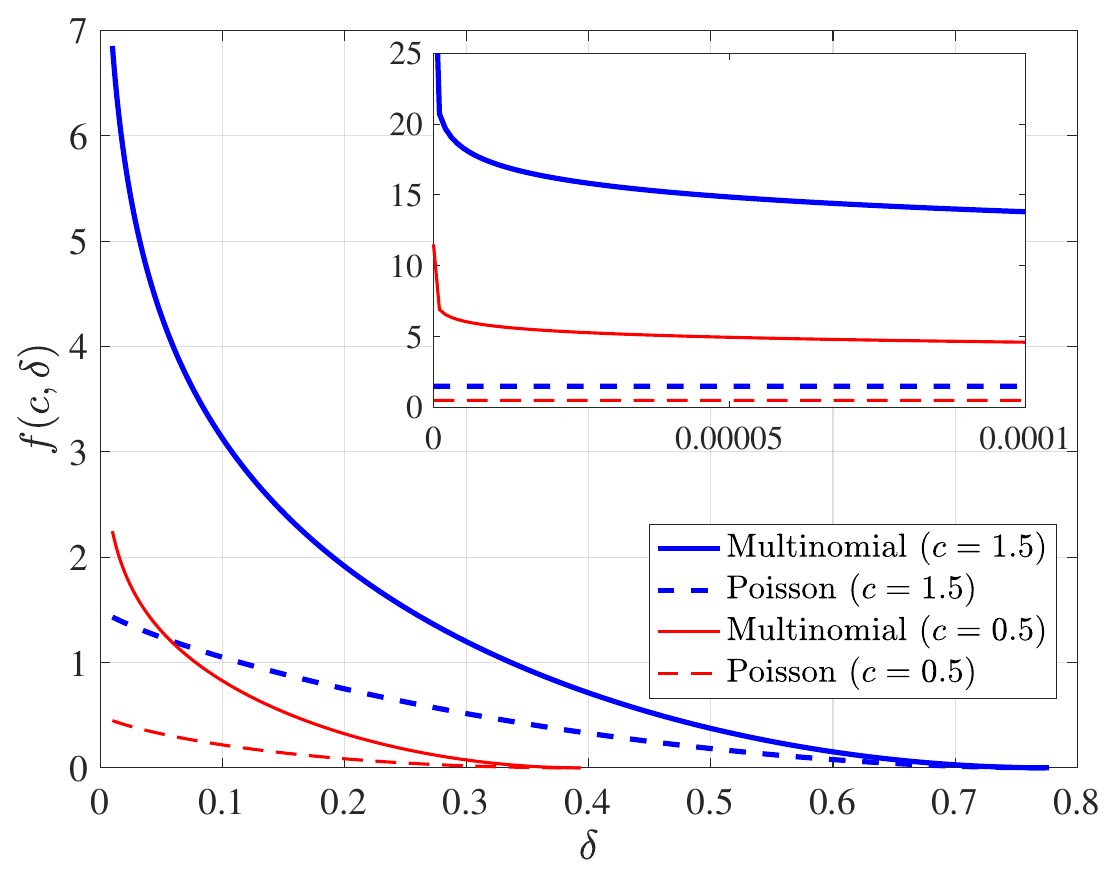}
    
    \caption{A plot of $f(c,\delta)$ against $\delta$ for the multinomial and Poisson sampling models, with $c = 0.5$ and $c = 1.5$.  Note that the multinomial curves approach $\infty$ as $\delta \to 0$, but this is only visible for extremely small $\delta$ values, as we see in the zoomed part of the plot.} \vspace*{-3ex}
    \label{fig:poisson_vs_multinomial}
\end{figure}

The fact that our error exponent grows unbounded as $\delta \to 0$ motivates the study of ``low-rate'' scaling regimes, where the number of messages is $e^{o(M \log M)}$.  The above discussion indicates that we should expect $e^{-\omega(M)}$ decay in the error probability, but does not provide the precise scaling.  This will be addressed in Sections \ref{sec:lowrate_no_seq} and \ref{sec:lowrate_seq}.


\subsection{Discussion of the Proof of Theorem \ref{thm:main_result}}

Before proceeding with the proof of Theorem \ref{thm:main_result}, we pause to highlight a key principle that we found to be particularly useful, not only for Theorem \ref{thm:main_result} but also for the additional results to come in Sections \ref{sec:superlinear}--\ref{sec:lowrate_seq} (on the regime $N = \omega(M)$ and certain low-rate regimes). 

Suppose that the true message is $i$ (with outer codeword $A_i$), and let $j$ be some incorrect message.  An important source of error is that if the $N$ sampled molecules (before sequencing errors) all lie in $A_i \cap A_j$, then distinguishing $i$ and $j$ becomes impossible.  Naturally, the overall error probability (given message $i$) is the union of error events across \emph{all} $j \ne i$.  For the achievability part, it is tempting to use a union bound so that we can focus on a single $j$ at a time.  Moreover, it is well-known that the union bound (truncated to 1) is tight to with a factor of 2 for independent events (e.g., see \cite[Lemma A.2]{shulman2003communication}).

However, the error events associated with multiple $j$ values are in fact \emph{far from independent}, and our approach is based on the observation that they are \emph{dependent via the number of distinct molecules sampled} (denoted by $K$) -- a smaller $K$ value implies that errors are more likely to occur.  Accordingly, our analysis is based on understanding the random behavior of $K$, and the quantity $p(N,M,K)$ in \eqref{eq:K_def} thus naturally arises.  Intuitively, the case that $K$ is ``too small'' serves as an ``outage event'' that prevents reliable recovery of the message, and this is what dominates the overall error probability.   We found this approach to permit a relatively simple analysis (at least in constant-rate scaling regimes) while giving the optimal exponent.


\section{Proof of Theorem \ref{thm:main_result}} \label{sec:const_rate}

\subsection{Converse Bound} \label{sec:conv_isit}

We will prove the converse bound by assuming that there are no sequencing errors, i.e., whenever a molecule is sampled, it is observed perfectly without noise.  This is an easier problem than the original one (since noise could always be artificially introduced), so any converse still remains valid.

We consider a genie-aided argument inspired by others that have been used previously (e.g., see \cite[Sec.~3.2.2]{fundamental_limit_dna_2017}).  Specifically, we suppose that for each molecule $y$ received by the decoder, the decoder is told the multiplicity of $y$ in the encoder input.  For the concatenated coding based class that we consider, the decoder can compile this information into a partial frequency vector $v$ of length $M^{\alpha}$ (i.e., one entry for each molecule in the inner codebook):
\begin{itemize}
    \item If a molecule $y$ is received at least once by the decoder, then $v_y$ equals the multiplicity of $y$ in the input molecules.
    \item Otherwise, $v_y = 0$.
\end{itemize}
Observe that all the information relevant for estimating $m$ available to the decoder is captured by $v_y$.  This is because (i) the set of molecules observed (with duplicates removed) precisely matches the set of coordinates with $v_y>0$, and (ii) since sampling is invariant to the ordering of input molecules, seeing the same molecule $y$ multiple times does not reveal additional information beyond its multiplicity in the input set (which the genie already provides).

Let $\supp(v)$ denote the set of all coordinates $i$ with $v_i > 0$, and let $\|v\|_0 = |\supp(v)|$, which is equal to the number of distinct molecules seen by the decoder. 
Fix $\delta < \rbar$, and for each message $m \in \{1,2,\dotsc,\exp(RML)\}$, define
\begin{equation}
\hat{p}(m) = \mathbb{P}( \|v\|_0 \leq \delta M \mid m). 
\end{equation}
We momentarily consider a hypothetical scenario in which the input molecules $\{x_i\}_{i=1}^{M}$ are tagged with their respective indices as $\{(x_i, i)\}_{i=1}^M$. Let $\tilde{N}$ denote the number of distinct tagged molecules seen by the decoder (i.e., the number of $(x_i,i)$ pairs that get sampled at least once).  Since any collection of distinct molecules is also also a collection of distinct tagged molecules but not necessarily vice versa, we have $\tilde{N} \geq \|v\|_0$, and hence
\begin{equation}
    \hat{p}(m) \geq \mathbb{P}(\tilde{N} \leq \delta M | m) = p(cM, M, \delta M), \label{eq:Ntilde}
\end{equation}
where the second equality follows from the definition of $p(\cdot,\cdot,\cdot)$ in \eqref{eq:def_p}.

Since we have established that $v$ captures all relevant information for estimating $m$, we can treat the decoder as operating directly on $v$.  In addition, by Yao's minimax principle \cite[Sec.~2.2.2]{Mot10}, it suffices to consider deterministic decoders, so that the decoder's estimate is a deterministic function of $v$, which we denote by $g(v)$.  Define
\begin{equation}
    W = \{ g(v) \mid \|v\|_0 \leq \delta M\}.
\end{equation}
Observe that with $m$ being the true message, if $m\notin W$ and $\|v\|_0 \leq \delta M$, then $g(v) \neq m$, meaning that a failure occurs. Thus, and by \eqref{eq:Ntilde}, the error probability $P_e$ satisfies
\begin{equation}
    P_e \geq \mathbb{P}(m \notin W) \cdot p(cM, M,\delta M)
    \label{eq:pe_w}
\end{equation}
with $m$ drawn uniformly at random from $\{1,2,\dotsc,\exp(RML)\}$. 
We now proceed to bound $|W|$, which is at most the number of possible $v$ with $\|v\|_0 \leq \delta M$.

Recall that we defined $\alpha > 1$ such that there are $M^{\alpha}$ codewords in the inner code.  
The total number of choices for $\supp(v)$ is simply the number of non-empty subsets of $\{1,2,\ldots, M^\alpha\}$ with size at most $\delta M$, i.e., $\sum_{i=1}^{\delta M} \binom{M^\alpha}{i}$. For large enough $M$, we have $\delta M \leq \frac12 M^\alpha$ (since $\alpha > 1$), and therefore $\binom{M^\alpha}{i} \leq \binom{M^\alpha}{\delta M}$ for all $1 \leq i \leq \delta M$. Therefore, the number of possible choices for $\supp(v)$ is at most 
\begin{equation}
    \sum_{i=1}^{\delta M} \binom{M^\alpha}{i} \leq (\delta M) \binom{M^\alpha}{\delta M} \leq (\delta M)(e M^{\alpha-1}/\delta)^{\delta M}.
\end{equation}
Moreover, if we fix a choice for $\supp(v)$, then there are 
at most $2^M$ choices for $v$.\footnote{This is because any such $v$ can uniquely be mapped to a length-$M$ binary sequence $(0\dotsc01) \circ (0\dotsc01) \circ \dotsc \circ (0\dotsc01)$, where $\circ$ denotes string concatenation and the length of the $i$-th segment $0\dotsc01$ is equal to the $i$-th non-zero value of $v$ (for $i=1,\dotsc,\|v\|_0$).} Since each element of $W$ must equal $g(v)$ for at least one $v$ with $\|v\|_0 \leq \delta M$, it follows that 
\begin{align}
    |W| & \leq (\delta M)(e M^{\alpha-1}/\delta)^{\delta M} \cdot 2^M \\&= \exp((\alpha-1) \delta M \log M) \cdot \mathcal{O}(1)^M.
    \label{eq:upper_bound_v0}
\end{align}

We now return to \eqref{eq:pe_w}, in which $m$ is chosen uniformly random from $\{1,2,\ldots, \exp(RML)\}$.  Recalling from \eqref{eq:RML} that $\exp(RML) = \exp((\alpha-1)\rbar M \log M)$, we have
\begin{align}
    \mathbb{P}(m \in W)
     &= \frac{|W|}{\exp((\alpha-1)\rbar M \log M)}\\ &\leq \exp((\alpha-1)(\delta-\rbar) M \log M) \cdot \mathcal{O}(1)^M. \label{eq:pe_w_bound} 
\end{align}
Since $\delta < \rbar$, it follows that $\mathbb{P}(m\in W) \rightarrow 0$ and thus $\mathbb{P}(m \notin W) \rightarrow 1$.
Combining this with \eqref{eq:pe_w} and the definition of $f(c,\delta)$ (see Theorem \ref{thm:f_formula}) then gives
$\lim_{M\rightarrow \infty} -\frac{1}{M}\log P_e  \le f(c,\delta)$.  
Since this holds for all $\delta < \rbar$ and $f$ is continuous, we deduce the desired bound, i.e., $\lim_{M\rightarrow \infty} -\frac1M \log P_e \leq  f(c,\rbar)$.

\subsection{Achievability Bound} \label{sec:ach}

In the achievability part, we need to allow for sequencing errors.  Since we assume that the inner rate $\rin$ is achievable, we know that there exists a codebook with $o(1)$ error probability in each invocation of sequencing.  We show that under this assumption alone, the error exponent $f(c,\rbar)$ is achievable.  Our encoding strategy is index-based (see Definition \ref{dfn:index_based}), which in turn implies that the $M$ molecules in any given outer codeword are all distinct.

\subsubsection{Decoding rule}

Recall that the outer codebook is $(A_i)_{i=1}^{\exp(RML)}$,
where the encoder stores the subset of molecules $A_i$ upon receiving message $i$. 
Let $S$ be the set of molecules (with duplicates removed) produced by applying the inner decoder $D(\cdot)$ to the received sequences $(y_1,y_2,\dotsc,y_N)$. We consider an outer decoder that simply chooses $i$ to maximize $|S \cap A_i|$:
\begin{equation}
    \hat{m} = {\rm argmax}_{i=1,\dotsc,e^{RML}} |S \cap A_i|. \label{eq:dec_rule}
\end{equation}

\subsubsection{A sufficient condition for decoding to succeed}

We first establish sufficient conditions for success.

\begin{lemma}
Let $K$ be the number of distinct molecules sampled, and let $T$ be the number of sequencing errors (i.e., cases where some $x$ is sampled to obtain $y$ but $D(y) \ne x$).  Then the decoder succeeds provided that, for some $\epsilon > 0$, the following conditions hold:
\begin{itemize}
    \item[(i)] $T \leq \epsilon M$;
    \item[(ii)] $K \geq (\rbar + 3 \epsilon) M$;
    \item[(iii)] $|A_i \cap A_j| < (\rbar + \epsilon) M$ for all $i \neq j$.
\end{itemize}
\label{lem:sufficient_decoding_success}
\end{lemma}

\begin{proof}
Let $A_i$ be the codeword (containing $M$ molecules) chosen by the encoder, and let $S_0$ be the set of molecules sampled (with duplicates removed) before sequencing errors. 
For each $x \in S_0 \setminus S$, there must be at least one sequencing error that replaced $x$ by something else, and thus $|S_0\setminus S| \leq T$.  
Similarly, for each $x \in S \setminus S_0$, there must be at least one sequencing error that replaced another molecule with $x$, and thus $|S\setminus S_0| \leq T$.  Condition (i) in the lemma statement thus gives $|S_0\setminus S| \le \epsilon M$ and $|S\setminus S_0| \leq \epsilon M$.

Since $S_0 \subseteq A_i$, we have $|S_0\cap A_i|=|S_0|=K$, which is at least $(R_0 + 3\epsilon)M$ by condition (ii), while for all $j\neq i$, $|S_0\cap A_j| \leq |A_i\cap A_j| < (\rbar +\epsilon)M$ by condition (iii). 
We therefore conclude that
\begin{gather}
    |S\cap A_i| \geq |S_0\cap A_i| - |S_0 \setminus S| \geq (\rbar + 2\epsilon) M, \\ 
    |S\cap A_j| \leq |S_0 \cap A_j| + |S \setminus S_0| < (\rbar + 2\epsilon) M,
\end{gather}
and thus the decoding rule \eqref{eq:dec_rule} is successful.
\end{proof}

\subsubsection{Existence of good codebooks}

Next, we give a construction that ensures the ``well-separated'' property in condition (iii) of Lemma \ref{lem:sufficient_decoding_success}. 

\begin{lemma}
    Fix $\epsilon>0$, and consider any inner codebook of size $M^{\alpha}$ (with all codewords being distinct). For all sufficiently large $M$, there exists an index-based outer codebook of size $\exp((\alpha-1)\rbar M \log M)$ (as per \eqref{eq:RML}) such that for all $i\neq j$,
\begin{equation}
    |A_i \cap A_j| < (\rbar + \epsilon) M.
    \label{eq:collision}
\end{equation}
\label{lem:collision_free}
\end{lemma}
\begin{proof}
The proof closely resembles the classical Gilbert-Varshamov construction  (e.g., see \cite[Ex.~5.19]{gallager1968information}).
Given the $M^\alpha$ molecules in the inner code, we arbitrarily arrange them into $M$ groups of size $M^{\alpha-1}$.  All of our (outer) codewords will contain exactly one molecule from each group, so that our codebook is index-based according to Definition \ref{dfn:index_based}, and the $M$ molecules comprising each codeword are distinct.  Subsequently, we let $A$ represent a generic candidate codeword.

We construct a codebook using a naive greedy argument: Simply add more codewords while preserving \eqref{eq:collision} until it is impossible to do so further.  
Accordingly, we say that a candidate codeword $A$ is \emph{blocked} if there exists some previously selected $A_i$ for which $|A_i \cap A| \ge (\rbar+\epsilon) M$. After $A_1, A_2, \ldots, A_i$ are chosen, we simply choose $A_{i+1} = A$ for some arbitrary $A$ that has not been blocked, and continue until every set is blocked.

For any specific $A_i$, the number of $A$ such that $|A_i \cap A| \ge (\rbar+\epsilon) M$ is bounded above by
\begin{equation}
    2^M M^{(\alpha-1)(1-\rbar-\epsilon)M}. \label{eq:num_blocked}
\end{equation}
This is because every such set $A$ can be described by $(A\cap A_i, A \setminus A_i)$; there are at most $2^M$ choices for $|A \cap A_i|$ (since $|A_i| = M$), and after $A\cap A_i$ is chosen, there are $|A\setminus A_i| \leq (1-\rbar-\epsilon)M$ more molecules to choose (since $|A \setminus A_i| = |A| - |A_i \cap A|$), each with $M^{\alpha-1}$ choices.

For index-based codes, there are a total of $(M^{\alpha-1})^M = M^{(\alpha-1)M}$ possible codewords, so in order for all codewords to be blocked, the codebook must contain at least
\begin{equation}
    \frac{M^{(\alpha-1)M}}{2^M M^{(\alpha-1)(1-\rbar-\epsilon)M}} = 2^{-M} \cdot M^{(\alpha-1)(\rbar+\epsilon) M} \label{eq:num_chosen}
\end{equation}
codewords. 
Since $\alpha > 1$, for sufficiently large $M$, we have $2^{-M} \cdot M^{(\alpha-1)\epsilon M} > 1$ 
so that the total number of codewords is at least $M^{(\alpha-1)\rbar M}$, 
and therefore Lemma \ref{lem:collision_free} follows.
\end{proof}

\subsubsection{Completing the achievability proof}

Fix $\epsilon>0$, and recall that $T$ is the number of sequencing errors.  We proceed to characterize the probabilities of conditions (i) and (ii) in Lemma \ref{lem:sufficient_decoding_success} occurring.

Condition (i) concerns the event $T \geq \epsilon M$.  Using the assumption that $\rin$ is achieved, the probability of a specific sequencing error occurring approaches zero as $L \to \infty$, i.e., it behaves as $o(1)$. 
Whenever $T \geq \epsilon M$, there exists a set of indices $\mathcal{I} \subseteq \{1,2\ldots,N\}$ with $|\mathcal{I}| \geq \epsilon M$ in which for each $k \in \mathcal{I}$, a sequencing error occurs in the $k$-th out of the $N$ samples. Taking a union bound over all such $\mathcal{I}$ and using $N = cM$ gives
\begin{align}
    \mathbb{P}(T \geq \epsilon M) &\leq 2^N \cdot \left( o(1) \right)^{\epsilon M}\\
    & = \exp(-\omega(M)), \label{eq:p_super_exp}
\end{align}
i.e., the decay to zero is faster than exponential.

Condition (ii) concerns the event $K \geq (\rbar + 3 \epsilon) M$, where $K$ is defined in the statement of Lemma \ref{lem:sufficient_decoding_success}.  Combining this definition with that of $p(\cdot,\cdot,\cdot)$ in \eqref{eq:def_p} gives 
\begin{equation}
    \mathbb{P}(K \leq (\rbar + 3\epsilon) M) = p(cM, M, (\rbar + 3\epsilon) M),
\end{equation} so that the definition of $f$ (see Theorem \ref{thm:f_formula}) gives
\begin{equation}
    \lim_{M \rightarrow \infty} -\frac1M \log \mathbb{P}(K \leq (\rbar + 3\epsilon) M) = f(c, \rbar + 3\epsilon).
\end{equation}

Applying Lemma \ref{lem:sufficient_decoding_success} and using the codebook from Lemma \ref{lem:collision_free}, the error probability $P_e$ is upper bounded by
\begin{equation}
    P_e \leq \mathbb{P}(T \geq \epsilon M) + \mathbb{P}(K  \geq (\rbar + 3\epsilon) M).
\end{equation}
The exponentially decaying term clearly dominates the super-exponentially decaying one, and we deduce that
\begin{eqnarray}
    \lim_{M \rightarrow \infty} -\frac1M \log P_e \geq f(c, \rbar + 3\epsilon).
\end{eqnarray}
Since this holds for all $\epsilon>0$, the continuity of $f$ gives
\begin{equation}
    \lim_{M \rightarrow \infty} -\frac1M \log P_e  \geq f(c,\rbar)
\end{equation}
as desired.

%% file: appendix.tex
\section{Proofs of Theorem \ref{thm:f_formula} (Balls and Bins)} \label{app:balls_bins_pf}

Regarding the second part of Theorem \ref{thm:f_formula}, note that when $\delta = 1-\exp(-c)$, we have $r=1$ in \eqref{eq:r}, and the right-hand side of \eqref{eq:f_eval} is zero:
\begin{equation}
     -c\log r - H_2(\delta) + r H_2\Big(\frac{\delta}{r}\Big)= 0. \label{eq:f_eval_0}
\end{equation}
If $\delta > 1-\exp(-c)$, then the monotonicity of $p$ in its third argument (which follows directly from its definition) gives
\begin{align}
    0 &\leq -\frac1M\log p(cM, M, \delta M) \\ &\leq -\frac1M\log p(cM, M, (1-\exp(-c)) M).
\end{align}
Thus, if we prove Theorem \ref{thm:f_formula} for $\delta = 1-\exp(-c)$, then the squeeze theorem gives for all $\delta > 1-\exp(-c)$ that
\begin{equation}
    f(c,\delta) = \lim_{M \rightarrow \infty} -\frac1M\log p(cM, M, \delta M) = 0,
\end{equation}
thus also establishing the theorem for all such cases.  As such, we will focus on the case that 
\begin{equation}
    1-\exp(-c) \geq \delta.  \label{eq:focus_case}
\end{equation}
We proceed to establish that $r$ is well-defined, and then derive matching upper and lower bounds that combine to give the theorem.

\subsection{Existence and uniqueness of $r$ in \eqref{eq:r}}

Since the theorem states that $r$ should be a unique number in $(\delta,1]$ satisfying \eqref{eq:r}, we proceed to understand the endpoints $\delta$ and $1$, as well as a useful monotonicity property in between them.

In accordance with the theorem statement, we are interested in the function $\psi(x) = x(1-\exp(-c/x))$.  Using $1+\frac{c}{x} <\exp(c/x)$ for $c,x > 0$, we have
\begin{equation}
    \frac{d}{dx} x(1-\exp(-c/x)) = 1-e^{-c/x}\Big(1+\frac{c}{x}\Big)>0 \label{eq:derivative}
\end{equation}
so that $\psi$ is strictly increasing with respect to $x>0$. When $x=\delta$, we have
\begin{equation}
     x(1-\exp(-c/x)) < \delta,
\end{equation}
and when $x=1$, we have
\begin{equation}
    x(1-\exp(-c/x)) = 1-\exp(-c) \geq \delta
\end{equation}
by \eqref{eq:focus_case}.  
These conditions guarantee the uniqueness and existence of a root $r \in (\delta,1]$ for which $\psi(r)=\delta$, as desired.

The bounded derivative in \eqref{eq:derivative} further implies that $r$ is a continuous function of $\delta$ (for fixed $c$). Since the right-hand side of \eqref{eq:f_eval} is continuous with respect to $r$, this also shows that $f$ is continuous with respect to $\delta$, as stated in Theorem \ref{thm:f_formula}.

\subsection{Upper bound}
Let $q(N,K)$ denote the probability (possibly 0) that we throw $N$ balls into $K$ bins and all $K$ bins are non-empty (we will later set $K=\delta M$ and $N=cM$). We claim that
\begin{equation}
    p(N,M,K) = \sum_{i=0}^K \binom{M}{i} \left(\frac{i}{M}\right)^N q(N,i).
    \label{eq:cumulative_p}
\end{equation}
To see this, let $S$ be any set of $i$ bins, noting that there are $\binom{M}{i}$ such sets $S$. The probability that every ball lands in $S$ is $(\frac{i}{M})^N$. Conditioned on every ball landing in $S$, the probability that every bin in $S$ is non-empty is given by $q(N,i)$. Multiplying these three quantities together gives the probability that exactly $i$ bins are non-empty, and taking the sum over $i=0$ to $K$ gives \eqref{eq:cumulative_p}.

We fix an integer $M_0 \in [K,M]$, and consider the ratio $\frac{\binom{M}{i}}{\binom{M_0}{i}}$ for $i \le K$.  Since each term of the form $\frac{M-j+1}{M_0-j+1}$ ($0\leq j \leq K$) is larger than 1, we find that 
\begin{align}
    \frac{\binom{M}{i}}{\binom{M_0}{i}} &= \frac{M(M-1)\ldots (M-i+1)}{M_0(M_0-1)\ldots (M_0-i+1)} \nonumber \\
    & \leq \frac{M(M-1)\ldots (M-K+1)}{M_0(M_0-1)\ldots (M_0-K+1)} = \frac{\binom{M}{K}}{\binom{M_0}{K}}.
    \label{eq:binom_ratio}
\end{align}
Therefore, for all $M_0 \in [K,M]$, we have
\begin{align}
    &p(N,M,K)\\ &\stackrel{\eqref{eq:cumulative_p}}{=} \sum_{i=0}^K \binom{M}{i} \left(\frac{i}{M}\right)^N q(N,i)\\
    & \stackrel{\eqref{eq:binom_ratio}}{\leq}  \frac{\binom{M}{K}}{\binom{M_0}{K}} \left(\frac{M_0}{M}\right)^N \sum_{i=0}^K \binom{M_0}{i} \left(\frac{i}{M_0}\right)^N q(N,i)\\
    &\stackrel{\eqref{eq:cumulative_p}}{=}  \frac{\binom{M}{K}}{\binom{M_0}{K}} \left(\frac{M_0}{M}\right)^N p(N, M_0, K)\\
    & \leq  \frac{\binom{M}{K}}{\binom{M_0}{K}} \left(\frac{M_0}{M}\right)^N.
\end{align}
We now substitute $M_0 = rM$, $K = \delta M$, and $N = cM$; since $r \in (\delta,1]$, these choices are consistent with the assumption $M_0 \in [K,M]$.  With these substitutions, we have
\begin{equation}
    -\frac1M \log p(cM,M,\delta M) \geq -\frac1M \left( cM \log r + \log \frac{\binom{M}{\delta M}}{\binom{r M}{\delta M}}\right).
\end{equation}
Using the fact that $\log \binom{a}{b} = a H_2(b/a) + \mathcal{O}(\log a)$, we obtain
\begin{multline}
    \lim_{M \to \infty} -\frac1M \log p(cM,M,\delta M) \\ \geq -c\log r - H_2(\delta) + r H_2\Big(\frac{\delta}{r}\Big), \label{eq:to_match}
\end{multline}
where the use of lim instead of liminf/limsup will be justified by the subsequent matching lower bound on $p(cM,M,\delta M)$ (i.e., an upper bound on $-\frac{1}{M} \log p(cM,M,\delta M)$).

\subsection{Lower bound}
We perform a similar computation as the upper bound. Fix $\epsilon > 0$, and again let $M_0 = rM$. Similar to \eqref{eq:binom_ratio}, whenever $K-\epsilon M \leq i \leq K \leq M_0 \leq M$, we have
\begin{equation}
    \frac{\binom{M}{i}}{\binom{M_0}{i}} \geq \frac{\binom{M}{K - \epsilon M}}{\binom{M_0}{K-\epsilon M}}, \label{eq:binom_ratio2}
\end{equation}
so that
\begin{align}
     &p(N,M,K) \nonumber \\ &\stackrel{\eqref{eq:cumulative_p}}{\geq} \sum_{i=K -\epsilon M}^K \binom{M}{i} \left(\frac{i}{M}\right)^N q(N,i)\\
    &\stackrel{\eqref{eq:binom_ratio2}}{\geq} \frac{M_0^N}{M^N} \frac{\binom{M}{K-\epsilon M}}{\binom{M_0}{K-\epsilon K}}\sum_{i=K -\epsilon M}^K \binom{M_0}{i} \left(\frac{i}{M_0}\right)^N q(N,i)\\
    &\stackrel{\eqref{eq:cumulative_p}}{=} (p(N,M_0,K) - p(N,M_0,K-\epsilon M)) \frac{M_0^N}{M^N} \cdot \frac{\binom{M}{K-\epsilon M}}{\binom{M_0}{K-\epsilon M}}.\label{eq:p_lower}
\end{align}
We now prove the following.

\begin{lemma}
    Under the preceding setup with $M_0 = rM$, $N=cM$, and $K=\delta M$, it holds that 
    \begin{equation}
        p(N,M_0,K) - p(N,M_0,K-\epsilon M) = \Omega(1/M).
        \label{eq:p_lower_1m}
    \end{equation}
\end{lemma}
\begin{proof}
We first characterize the expected number of non-empty bins in the case that there are $N$ balls and $M_0$ bins. Any given bin is empty with probability $(\frac{M_0-1}{M_0})^N < \exp(-N/M_0)$.  We apply linearity of expectation to conclude that the expected number of non-empty bins is at least 
\begin{equation}
    {M_0} (1-\exp(-N/M_0)) = K, \label{eq:M0N0K}
\end{equation}
where the equality follows from \eqref{eq:r} along with $M_0 = rM$, $N=cM$, and $K=\delta M$.

Define the random variables $X_1, X_2,\ldots, X_N$ such that each $X_i$ is the index of the bin that the $i$-th ball lands in.  Then, the number of non-empty bins can be written as $g(X_1,X_2,\ldots,X_N)$, where $g$ is the function that outputs the number of distinct elements.  
If we change one value of $X_i$, $g$ changes at most by 1. Thus, we may apply McDiarmid's inequality \cite[Sec.~6.1]{Bou13} to obtain
\begin{equation}
    p(N,M_0,K-\epsilon M) \leq \exp\left( -\frac{2(\epsilon M)^2}{N}\right),
\end{equation}
which is exponentially small.

It remains to show that $p(N,M_0, K) = \Omega(1/M)$, and accordingly, we again consider $N$ balls and $M_0$ bins.  
If we throw $N-1$ balls and they result in $K-1$ or fewer non-empty bins, then after the $N$-th ball is thrown, the total number of non-empty bins is at most $K$. Moreover, if we see exactly $K$ non-empty bins after throwing $N-1$ balls, and the last ball lands in one of these $K$ non-empty bins, then the total number of non-empty bins will be exactly $K$. It follows that
\begin{equation}
    p(N,M_0,K) \geq p(N-1, M_0, K) \cdot \frac{K}{M_0}.
    \label{eq:reduce1}
\end{equation}
For any $N_0 < N$, repeatedly applying \eqref{eq:reduce1} (which is true for all values of $N$) gives
\begin{equation}
    p(N,M_0,K) \geq p(N_0, M_0, K) \cdot \Big(\frac{K}{M_0}\Big)^{N-N_0}. \label{eq:N0_recursed}
\end{equation}
We proceed under the choice $N_0 = N - \lceil c/r \rceil$. 

Since $(\frac{K}{M_0})^{N-N_0} = \big(\frac{\delta}{r}\big)^{\lceil c/r \rceil}$ is constant, it suffices to show that $p(N_0, M_0, K) = \Omega(1/M)$. With $N_0$ balls and $M_0$ bins, the probability of a specific bin being empty is
    $$\Big(\frac{M_0-1}{M_0}\Big)^{N_0} > \exp\Big(-\frac{N_0}{M_0-1}\Big) \geq \exp\Big(-\frac{N}{M_0}\Big),$$
where the first inequality follows by taking the reciprocal on both sides of $1+\frac{1}{M_0-1} < \exp(\frac{1}{M_0-1})$, and the second inequality follows since $\frac{N_0}{M_0-1} \le \frac{N - c/r}{ M_0 - 1 } = \frac{c(M-1/r)}{r(M-1/r)} = \frac{c}{r} = \frac{N_0}{M_0}$ (recall the choices $N = cM$ and $M_0 = rM$).  Hence, the expected number of non-empty bins is at least $M_0(1-\exp(-\frac{N}{M_0})) = K$ (see \eqref{eq:M0N0K}).

Then, by Markov's inequality, the probability of seeing at least $K+1$ non-empty bins is at most $\frac{K}{K+1}$, 
or equivalently, the probability of seeing at most $K$ non-empty bins is at least $\frac1{K+1}$. We therefore get
\begin{eqnarray}
    p(N,M_0,K) &\stackrel{\eqref{eq:N0_recursed}}{\geq} &p(N_0, M_0, K) \cdot \Big(\frac{K}{M_0}\Big)^{N-N_0}\\
    &\geq &\frac1{K+1} \Big(\frac{K}{M_0}\Big)^{N-N_0},
\end{eqnarray}
which scales as $\Omega\big(\frac{1}{M})$ since $(\frac{K}{M_0})^{N-N_0} = \big(\frac{\delta}{r}\big)^{\lceil c/r \rceil}$ is constant and $K = \delta M$.
\end{proof}

Combining \eqref{eq:p_lower} and \eqref{eq:p_lower_1m}, we obtain
\begin{equation}
    p(N,M,K) \geq \frac{M_0^N}{M^N} \cdot \frac{\binom{M}{K-\epsilon K}}{\binom{M_0}{K-\epsilon K}} \cdot \Omega\Big(\frac{1}{M}\Big).
\end{equation}
By the same argument as the upper bound above, the exponent associated with the right-hand side is
\begin{equation}
-c\log r - H_2(\delta-\epsilon) + r H_2\Big(\frac{\delta-\epsilon}{r}\Big).
\end{equation}
Therefore for all $\epsilon > 0$,
\begin{multline}
   \lim_{M\rightarrow \infty} -\frac1M \log p(cM, M, \delta M) \\ \leq -c\log r - H_2(\delta-\epsilon) + r H_2\Big(\frac{\delta-\epsilon}{r}\Big),
\end{multline}
and taking $\epsilon\rightarrow 0$ gives the required bound that matches \eqref{eq:to_match}.